%% file: main.tex
\documentclass{article}

\usepackage{arxiv2c}

\usepackage{amsmath,amssymb,amsfonts}

\usepackage{booktabs}   %
\usepackage{subcaption} %
\usepackage[utf8]{inputenc}

\input{commands}

\title{\myFrameworkName: A Verification Framework for Quantum Error Correction Codes}         %

\author{%
	Anbang Wu \\
	Department of Computer Science\\
	University of California, Santa Barbara \\
	\texttt{anbang@ucsb.edu} \\
	\And
	Gushu Li \\
	Department of Electrical \& Computer Engineering\\
	University of California, Santa Barbara \\
	\texttt{gushuli@ece.ucsb.edu} \\
	\AND
	Hezi Zhang \\
	Department of Computer Science \\
	University of California, Santa Barbara \\
	\texttt{hezi@ucsb.edu} \\
	\AND
	Gian Giacomo Guerreschi \\
	Intel Labs \\
	Santa Clara, California \\
	\texttt{gian.giacomo.guerreschi@intel.com}
	\AND
	Yuan Xie \\
	Department of Electrical \& Computer Engineering\\
	University of California, Santa Barbara \\
	\texttt{yuanxie@ucsb.edu} \\
	\AND
	Yufei Ding \\
	Department of Computer Science\\
	University of California, Santa Barbara \\
	\texttt{yufeiding@cs.ucsb.edu} \\
} 

\begin{document}
	
	\maketitle

\begin{abstract}

Quantum Error Correction (QEC) is essential for the functioning of large-scale fault-tolerant quantum computers, and its implementation is a very sophisticated process involving both quantum and classical hardware. Formulating and verifying the decomposition of logical operations into physical ones is a challenge in itself.
In this paper, we propose {\myFrameworkName}, a verification framework that can efficiently verify the formal correctness of stabilizer codes, arguably the most important class of QEC codes.
{\myFrameworkName} first comes with a concise language, {\langname},  where stabilizers are treated as a first-class object, to represent QEC programs.
Stabilizers are also used as predicates in our new assertion language, {\assnname}, as logical and arithmetic operations of stabilizers can be naturally defined.
We derive a sound quantum Hoare logic proof system with a set of inference rules for {\myFrameworkName} to efficiently reason about the correctness of QEC programs.
We demonstrate the effectiveness of {\myFrameworkName} with both theoretical complexity analysis and in-depth case studies of two well-known stabilizer QEC codes, the repetition code and the surface code.

\end{abstract}

\twocolumn
\input{introduction}

\input{background}

\input{ProgramSyntax}

\input{AssnLang}
\input{Evaluation}

\input{conclusion}

\bibliographystyle{unsrt}
\bibliography{bib}

\newpage
\appendix
\input{appendix}

\end{document}

%% file: commands.tex
\usepackage{amsthm}
\usepackage{color}
\usepackage{proof}
\usepackage{amsmath}
\usepackage{mathtools}
\usepackage{multirow}
\usepackage{caption}
\usepackage{subcaption}
\usepackage{thmtools}
\usepackage{tikz}
\usepackage{thm-restate}
\newcommand{\qstate}[1]{\ensuremath{\vert #1 \rangle}}
\newcommand{\vsep}{\ensuremath{\ \vert\ }}

\newcommand{\myFrameworkName}{\textit{QECV}}
\newcommand{\langname}{\textit{QECV-Lang}}
\newcommand{\assnname}{\textit{QECV-Assn}}

\newcommand{\svar}{\ensuremath{\mathcal{s}}}
\newtheorem{theorem}{Theorem}[section]
\newtheorem{proposition}[theorem]{Proposition}

\newtheorem{definition}[theorem]{Definition}
\newtheorem{lemma}[theorem]{Lemma}
\newtheorem{example}[theorem]{Example}
\newtheorem{program}{Program}[section]

\newcounter{cnt}
\newcommand\showcnt{\addtocounter{cnt}{1}\thecnt}
\newcommand{\conf}[1]{\ensuremath{\langle #1 \rangle}}
\newcommand{\denot}[1]{\ensuremath{[\![ #1 ]\!]}}
\newcommand{\ket}[1]{\ensuremath{\vert #1 \rangle}}
\newcommand{\bra}[1]{\ensuremath{\langle #1 \vert}}

\newcommand{\qif}[3]{\ensuremath{\textbf{if}\ M[#1]\ \textbf{then}\ #2\ \textbf{else}\ #3\ \textbf{end}}}
\newcommand{\stabnum}{\ensuremath{w}}
\newcommand{\qwhilelang}{quantum \textbf{while}-language}
\newcommand{\postpone}[1]{{Details are postponed to Appendix~{#1}}}
\newcommand{\myquad}{{\color{white}-}}
\newcommand{\qifn}[3]{\ensuremath{\textbf{if}\ M[#1]\ \textbf{then}\\\myquad #2\ \\\textbf{else}\\\myquad#3\ \\\textbf{end}}}

\newcommand{\prog}{\text{Prog}}
\newcommand{\projector}{P}

\usepackage{syntax}
\usepackage[many]{tcolorbox}
\newcommand{\myfontsize}{9pt}
\newcommand{\mylinesize}{9pt}
\newcommand{\myproof}[1]{{\noindent\textit{Proof}. #1 \hfill \qed}}
\newcommand{\nothmskip}{}

%% file: introduction.tex
\section{Introduction}

Quantum error correcting (QEC) codes~\cite{Fowler2012SurfaceCT, Chamberland2020BuildingAF, Chamberland2020TopologicalAS} are vital for implementing fault-tolerant quantum computation and overcoming the noise present in quantum hardware~\cite{Preskill2018QuantumCI, Holmes2020NISQBQ}. Quantum device vendors are exploring various quantum error correction codes to boost the error tolerance of quantum computation. For example, Google exploits the repetition code~\cite{nielsen2002quantum} to suppress errors in their Sycamore device~\cite{google50296}, IBM extends the surface code~\cite{Fowler2012SurfaceCT} to their low-degree superconducting quantum computers~\cite{Chamberland2020}, and Amazon utilizes the concatenated cat code~\cite{Chamberland2020BuildingAF} to build a fault-tolerant qubit.

A central concept in QEC code design is that of a ``stabilizer''~\cite{Gottesman1997StabilizerCA}. The term refers to a quantum operator which expresses the correlations present among the physical qubits forming the logical qubit. Operations to encode logical states or detect and correct errors can be derived once the stabilizers of the QEC code are provided.
In a stabilizer code, these primitive operations over logical qubits consist of quantum programs, one for each primitive.
As an example, Fowler et al.~\cite{Fowler2012SurfaceCT} developed a series of programs on the surface code to implement the primitive operations (e.g., a logical X gate, H gate, and CNOT gate) necessary for universal fault-tolerant quantum computation.
While executing these primitives, any stabilizer code implementation requires frequent measurements of the physical qubits to detect possible hardware errors and, thus, apply the appropriate correction operation.

When analyzing the correctness of a stabilizer code, there are two key aspects that need to be considered:
1) \textit{the correctness of the logical operation:} The stabilizer code must implement the desired logical operation over the logical qubits by applying several physical operations over the consituent physical qubits. %
2) \textit{the capability of error correction:} 
When hardware errors happen, there exist protocols for error decoding and correction which are based on the information extracted by measurements in QEC codes. %

To the best of our knowledge, there is no formal verification framework for QEC codes yet.
Previous works on quantum error correction~\cite{Fowler2012SurfaceCT, Chamberland2020TopologicalAS, Lao2020FaulttolerantQE, Chao2019FlagFE, Noh2020FaulttolerantBQ} demonstrate the correctness of the proposed QEC protocol by numerical simulation on QEC programs. 
However, this approach does not provide formal proof for the correctness of QEC codes.
We asked ourselves the question:

\begin{center}
    \textit{Can one formally verify QEC codes}\\
    \textit{using existing verification frameworks}\\
    \textit{for general quantum programs since QEC codes} \\
	\textit{are effectively a kind of quantum program?}
\end{center}

In this vein, one well-developed method for quantum program verification~\cite{ Wu2019FullstateQC,Zhou2019AnAQ, Li2020ProjectionbasedRA} is to use dynamic techniques such as quantum simulation. 
This category of methods can accurately characterize the quantum state evolution of small quantum programs but can not scale up to large quantum programs with more than 50 qubits due to the exponential computation overhead~\cite{Wu2019FullstateQC}.
This poor scalability of dynamic methods makes it inefficient for the verification of QEC codes since  
a reasonably fault-tolerant logical qubit would inevitably involve many physical qubits~\cite{Fowler2012SurfaceCT}.

Another type of verification works~\cite{Ying2012FloydhoareLF, Ying2018ReasoningAP, Unruh2019QuantumHL, DHondt2006QuantumWP,Selinger2004TowardsAQ, Feng2020QuantumHL, Feng2021VerificationOD}, exploits static analysis techniques to reason about quantum programs. 
These works all naturally incur exponential computation overhead since they need to track the evolution of some Hermitian matrices, which are of dimension $O(4^n)$ for a $n$-qubit system. 
Yu and Palsberg~\cite{Yu2021QuantumAI} recently proposed a computationally efficient quantum abstract interpretation technique to reason about the correctness for certain kinds of assertions.
Yet, trading in accuracy is not suitable for the verification of QEC codes which requires exact correctness.

Thus, our answer to the question above is:

\begin{center}
    \textit{No, adopting general verification frameworks }\\
    \textit{sacrifices either scalability or accuracy. } 
\end{center}

To this end, we build a formal verification framework crafted for quantum stabilizer codes to squeeze out the best verification efficiency without compromising accuracy.
Our approach rests on a central idea: while realizing scalable verification for a general quantum program is hard, it might be possible 
to efficiently verify the QEC codes
with a delicate separation between the hard and easy parts in the verification process.    
We observe
that most parts in verifying QEC codes turn out to fall into the easy region 
because they can be efficiently processed by preserving the high-level stabilizer information. 
In particular, stabilizers provide a compact description for QEC codes~\cite{Gottesman1997StabilizerCA}. Major components in QEC codes, e.g., error channels, Clifford gates, and parity measurements can all be described within the stabilizer formalism. Besides, it only takes $O(n^2)$ complexity to emulate Clifford operations on stabilizers~\cite{nielsen2002quantum}. Using stabilizers as predicates, we can potentially avoid unnecessary exponential computation overhead  in general quantum program reasoning.

We first propose a concise QEC programming language, {\langname}, where stabilizers are treated as first-class objects. 
This allows {\langname} to represent in an intuitive and compact form different operations in the QEC implementation, ranging from encoding to decoding and to error correction.
We develop operational semantics and denotational semantics for {\langname}, which lays the foundation for building up the syntax-directed verification system. 
One key enabler for our semantics design is the separation between quantum states of the physical qubits and the information captured by the stabilizers. 
It describes the former with partial density matrices and treats the latter as a classic program state. 
It significantly simplifies the computations associated with the stabilizers by avoiding the direct description of how the stabilizers are measured and instead focusing on how the high-level information is used in the decoding stage.

We further develop a new assertion language, named {\assnname}, in which the predicates are defined by stabilizers.
To enhance the logical expressive power, we introduce not only the standalone stabilizers, but also the arithmetic and logical expressions of them for expressing assertions.
The key insight behind such a design is to form a universal state space for verification, as a standalone stabilizer could not represent the whole state space. 
We also remark that, despite our {\langname} and {\assnname} are crafted for QEC codes, their design allows for broader applicability. Any quantum programs and quantum predicates that could be expressed in the {\qwhilelang}~\cite{Ying2012FloydhoareLF} and Hermitian-based predicates~\cite{DHondt2006QuantumWP} can also be %
expressed in our languages. With such an elegant property, our languages could potentially serve as the common foundation of both general quantum programs and QEC designs. It will therefore avoid the dilemma of choosing
between a general but less effective language or a domain-specific but more effective language.

Together with {\langname} and {\assnname}, we further establish a sound quantum Hoare logic for QEC programs. 
This proof system can demonstrate exponential time and space saving for most QEC operations (e.g., state preparation, Pauli gates, and error detection) when, in a real QEC program, the predicate formulae is commutable with the stabilizer variable in our quantum Hoare logic.
Even for the most challenging verification of the logical T gate implementation, our proof system may still have this strong advantage, depending on the actual T gate implementation of the target QEC code.

We give both theoretical analysis and detailed case studies for evaluating the proposed framework. 
We first compare {\myFrameworkName} with the vanilla quantum Hoare logic~\cite{Ying2012FloydhoareLF} and quantum \textbf{while}-language~\cite{Ying2012FloydhoareLF} in terms of the complexity when describing and verifying the quantum stabilizer codes. 
We then give very detailed, step-by-step case studies of two well-known QEC codes. This allows the reader to familiarize the concepts behind our framework and its usage for verifying the correctness of quantum stabilizer codes.

To summarize, our major contributions are as follows:
\begin{itemize}
    \item We propose {\langname}, a concise  programming language for QEC codes and give a full specification of its syntax and operational/denotational semantics. 
    \item We formulate a new assertion language {\assnname}. It is the first effort that exploits stabilizers for building universal assertions on quantum states.
    \item We develop a sound quantum Hoare logic framework based on {\assnname} and {\langname} with a set of inference rules to verify the correctness of QEC programs. 
    \item  We demonstrate the effectiveness of our framework with both theoretical complexity analysis and in-depth case studies of two well-known stabilizer QEC codes.
\end{itemize}

%% file: background.tex
\section{Background}
In this section, we introduce the background for our work.
We summarize our key notations in Table~\ref{tab:notations}.
We do not cover the basics of quantum computing (e.g.,  density operator, unitary transformation) and recommend~\cite{nielsen2002quantum} for reference.

\begin{table}[b]
\caption{Notation used in this paper.}\label{tab:notations}
\resizebox{0.45\textwidth}{!}{
\begin{tabular}{p{0.2\linewidth} | p{0.75\linewidth}}
$q$, $\bar{q}$ & qubits, a set of qubits, respectively; \\\hline
$\ket{\psi}$, $\ket{\phi}$, $\ket{0}$, $\ket{1}$, $\ket{+}$, $\ket{-}$ & pure quantum states; \\\hline
$\rho$, $\ket{\psi}\bra{\psi}$ & density operators; \\\hline
$U$ & unitary transformations; \\\hline
$O, M$ & Observable; Use $M$ to stress measurement;\\\hline
$\mathcal{H}$ & the Hilbert space of quantum states;  \\\hline
$\mathcal{D(H)}$ & the set of partial density operators on $\mathcal{H}$; \\
\end{tabular} }
\end{table}

\subsection{Quantum Error Correction and Stabilizer}

Most QEC codes consist of three stages: encoding, decoding, and error correction. The encoding protocol projects an unprotected state of the data qubits into the subspace generated by the logical states. The decoding protocol detects potential errors by performing parity measurements on data qubits. Finally, the correction protocol removes the errors by driving the quantum state back to the logical subspace.

\noindent\textbf{Stabilizer}. 
The stabilizer formalism proposed by Gottesman~\cite{Gottesman1997StabilizerCA} provides a unifying description of many QEC codes. 
Given a $n$-qubit state $\ket{\psi}$ and a Pauli string $s \in \otimes^n \{I,X,Y,Z\}$, %
 we say that $s$ is a \textit{stabilizer} of $\ket{\psi}$, or $\ket{\psi}$ is \textit{stabilized} by $s$, if $s\ket{\psi} = \ket{\psi}$.
 When the stabilized state $\ket{\psi}$ is clear in the context, we will simply say that $s$ is a stabilizer, without referring to the stabilized state. 
We then can use multiple stabilizers to naturally identify a subspace, which is the intersection of the stabilizers' projection subspaces, to represent the logical states.
Moreover, stabilizers can also be observables and they can be measured to ascertain whether the state of the data qubits is in the correct subspace. 
These measurements are called \textbf{stabilizer measurements}. 
Another advantage of the stabilizer formalism is that it can describe standard quantum error channels on stabilized states.
Once the stabilizers are determined, the QEC encoding, decoding, and error correction can be easily derived.

Due to the centrality of stabilizers in QEC codes, including them as fundamental concepts of the verification framework is a promising direction. 
Benefiting from the expressiveness of the stabilizer formalism, we believe that our verification framework for quantum stabilizer codes will be applicable to many existing QEC codes and will help designing novel implementations of fault-tolerant operations. 

In the writing time of this paper, we also notice that Rand, Sundaram, Singhal and Lackey~\cite{Rand2021StaticAO, Rand2021ExtendingGT, Rand2021GottesmanTF} develop an elegant type-checking system for 
general quantum programs based on the stabilizer formalism. %
Despite that we share same high-level insights in utilizing stabilizer formalism~\cite{Gottesman1997StabilizerCA}, we differ significantly in the overall optimization goal and the entire design framework. 
In addition, their work only consider quantum circuits and cannot deal with branch statements (e.g., if and while). 
These statements are indispensable for QEC programs. 
In contrast, we can handle branch statements by incorporating stabilizer variables in the design of quantum predicate logic. Last but not least, this paper also develops a compact language for QEC programs while their work follows the vanilla quantum circuit language~\cite{nielsen2002quantum}.

\subsection{Quantum Program Language}
The {\qwhilelang} proposed by Ying~\cite{Ying2012FloydhoareLF} provides a universal description of purely-quantum programs without classical variables. It focuses on characterizing basic quantum program structures
and its syntax is defined in Backus-Naur form as follows:

\begin{tcolorbox}[colback=yellow!10!white,
                  colframe=white!20!black]
\begin{grammar}
\label{equ:qwhile}
  \let\syntleft\relax
  \let\syntright\relax
<$\prog$> $\Coloneqq$ \textbf{skip}
\vsep $q \coloneqq \qstate{0}$
\vsep $\bar{q} \coloneqq U[\bar{q}]$
\vsep $\prog_1;\prog_2$
\alt $\textbf{case}\ M[\bar{q}] = \overline{m \to \prog_m}\ \textbf{end}$
\alt $\textbf{while}\  M[\bar{q}]=1\ \textbf{do}\ \prog_1\ \textbf{done}$
\end{grammar}
\end{tcolorbox}
Here $\prog$ plays the role of a statement of the {\qwhilelang}, but can also indicates the full program when seen as a sequence of statements. In the expressions above, $q$ denotes a quantum variable and $\bar{q}$ represents a quantum register associated with a finite number of quantum variables. The language constructs above are explained as follows: 
\setcounter{cnt}{0}
(\showcnt)  \textbf{skip} does nothing; 
(\showcnt) $q \coloneqq \qstate{0}$ prepare quantum variable $q$ in state $\qstate{0}$; 
(\showcnt) $\bar{q} \coloneqq U[\bar{q}]$ perform unitary operation $U$ on the quantum registers $\bar{q}$; 
(\showcnt) $\prog_1;\prog_2$ is the sequencing of statements; 
(\showcnt) $\textbf{case}\ M[\bar{q}] = \overline{m \to \prog_m}\ \textbf{end}$ measures the quantum variables in $\bar{q}$ with semi-positive Hermitian operators $M=\{M_0,M_1, \cdots, M_m\}$ and executes program $\prog_m$ if the measurement outcome is $m$; 
(\showcnt) $\textbf{while} \; M[\bar{q}]=1\ \textbf{do}\ \prog_1\ \textbf{done}$ measures qubits $\bar{q}$ and executes $\prog_1$ if the measurement outcome is 1. If the measurement outcome is 0, the while loop terminates. Here $M$ is assumed to have only two possible outcomes, $m=0,1$.

The semantics of {\qwhilelang} is developed by assuming partial density operators as quantum program states. 
The details  can be found in ~\cite{Ying2012FloydhoareLF}.

\subsection{Quantum Hoare Logic and Quantum Predicates}
\textit{Quantum Hoare logic}~\cite{Ying2012FloydhoareLF} provides a syntax-directed proof system for reasoning about quantum program correctness. The basic computing unit of quantum Hoare logic is the \textit{Hoare tripe}, which is in the form of $\{A\}\prog\{B\}$. Here $A$ and $B$ are \textit{quantum predicates}, $\prog$ is the quantum program. $A$ is often called \textit{precondition} while $B$ is called \textit{post-condition}. The general meaning of the quantum Hoare tripe is that, if the input state satisfies $A$, then the output state of $\prog$ satisfies $B$. The exact mathematical interpretation of the quantum Hoare tripe depends on the type of the predicates used.

A general quantum predicate is a Hermitian operator $O$~\cite{Ying2012FloydhoareLF} where $0 \le Tr(O\rho) \le 1, \forall \rho \in \mathcal{D(H)}$. %
A quantum state $\rho$ satisfies the quantum predicate $O$ depending on the value of $Tr(O\rho)$, which represents the expectation value of $O$ on state $\rho$. 
In practice, measuring $Tr(O\rho)$ can be very time-consuming . 
This can be avoided by restricting the predicate from a general Hermitian operator to a projection operator $\projector$ with the property that $\projector^2 = I$.
We say that a quantum state $\rho$ satisfies predicate $\projector$ (denoted by $\rho \models \projector$) if $\projector\rho = \rho$. %
A projection operator $\projector$ can also be described by its subspace with eigenvalue +1, namely $S_{\projector}= \{ \ket{\psi} s.t. \ket{\psi} = \projector\ket{\psi} \}$ in the qubits' Hilbert space $\mathcal{H}$. Birkhoff and Neumann~\cite{Birkhoff1936TheLO} define a quantum logic on the set of subspaces in $\mathcal{H}$ in which, for example, the logical {\bf and} corresponds to the intersection of subspaces. This construction induces logic operations %
on projection operators by considering the equivalent operation on the associated subspaces.

%% file: ProgramSyntax.tex
\section{{\langname}}
In this section, we introduce {\langname}, a concise language for QEC Programs. We define its syntax, operational semantics, and denotational semantics.

\subsection{Syntax}\label{subsect: syntax}
We restate the notation for quantum variables as follows:
Define $\text{qVar}$ as the set of quantum variables,  $q$ as a metavariable ranging over quantum variables, and $\bar{q}$ to be a quantum register associated with a finite set of distinct quantum variables. We denote the state space of $q$ by $\mathcal{H}_q$ which is a two-dimensional Hilbert space spanned by the computational basis states $\{\qstate{0}, \qstate{1}\}$. The state space of $\bar{q}$ is the tensor product of Hilbert spaces
$\mathcal{H}_{\bar{q}} = \otimes_{q\in \bar{q}}\mathcal{H}_q$.

Logical operations in QEC codes are often associated with changes in the set of stabilizer measurements. For example, the surface code~\cite{Fowler2012SurfaceCT} frequently turns on and turn off specific stabilizer measurement circuits to implement logical gates. Besides, the outcomes of stabilizer measurements act as signals for error correction. By introducing a stabilizer variable, which  represents a stabilizer measurement circuit without the need of specifying its actual implementation, we can greatly simplify the description of QEC programs. %
We define the notations for the stabilizer variable as follows:

Define $S$ as the set of stabilizers on $\text{qVar}$, $s$ as an individual stabilizer in $S$, $\text{sVar}$ as the set of stabilzer variables, and $\svar$ as a metavariable ranging over $\text{sVar}$. To avoid $S$ being uncountable, we assume that every $s \in S$ only involves a finite number of qubits. The range of values for the stabilizer variable $\svar$ is $S\cup -S \cup iS \cup -iS$, where $i$ is the imaginary unit.

We define the syntax of {\langname} as follows:
\begin{tcolorbox}[colback=yellow!10!white,
                  colframe=white!20!black]
\begin{grammar}

\label{equ:qeclang}
  \let\syntleft\relax
  \let\syntright\relax
<$\prog$> $\Coloneqq$ \textbf{skip}
\vsep $q \coloneqq \qstate{0}$
\vsep $\bar{q} \coloneqq U[\bar{q}]$
\vsep $\svar \coloneqq s_e^u$
\alt $\prog_1;\prog_2$
\alt $\textbf{if}\ M[\svar,\bar{q}]\ \textbf{then}\  \prog_1\ \textbf{else}\ \prog_0\ \textbf{end}$
\alt $\textbf{while}\  M[\svar,\bar{q}]\ \textbf{do}\ \prog_1\ \textbf{done}$ %

<$s_e^u$> $\Coloneqq$ $\pm s$ | $\pm i\,s$ | $\pm \svar$
\end{grammar}
\end{tcolorbox}

\setcounter{cnt}{0}
The proposed language constructs consisting of instructions as follows: 
(\showcnt)  \textbf{skip} does nothing; 
(\showcnt) $q \coloneqq \qstate{0}$ resets quantum variable $q$ to ground state $\qstate{0}$; 
(\showcnt) $\bar{q} \coloneqq U[\bar{q}]$ perform unitary operation $U$ on quantum register $\bar{q}$; 
(\showcnt) $\svar \coloneqq s_e^u$ assigns a unary stabilizer expression $s_e^u$ to the stabilizer variable $\svar$;
(\showcnt) $\prog_1;\prog_2$ is the sequencing of programs; 
(\showcnt) $\textbf{if}\ M[\svar,\bar{q}]\ \textbf{then}\  \prog_1\ \textbf{else}\ \prog_0\ \textbf{end}$ perform the stabilizer measurement represented by $\svar$ on qubits $\bar{q}$ (or in short, measures $\svar$ on qubits $\bar{q}$)
measures qubits in $\bar{q}$ with the stabilizer $\svar$ and executes program $\prog_1$ if the measurement outcome is $1$. If the measurement outcome is $-1$, i.e. \textbf{the measured state is in the -1 eigenspace of} $\svar$, we first execute $\svar=-\svar$ which flips the sign of $\svar$, then execute $\prog_0$. 
We can see (6) as a short-term for $\textbf{if}\ M[\svar,\bar{q}]=1\ \textbf{then}\ \prog_1\ \textbf{else}\ \svar \coloneqq -\svar ;\ \prog_0\ \textbf{end}$.
(\showcnt) $\textbf{while}\  M[\svar,\bar{q}]\ \textbf{do}\ \prog_1\ \textbf{done}$ measures $\svar$ on qubits $\bar{q}$, and perform $\prog_1$ if the measurement outcome is 1. If the measurement outcome is -1, the sign of $\svar$ will be flipped automatically, and then the while loop terminates.

The language constructs above are similar to those of the quantum \textbf{while}-language, except the part associated with stabilizer variables. Stabilizer variables can be used to describe operations on stabilizers.
For example, to turn off one stabilizer measurement circuit in a QEC program, we can simply set $\svar = {\rm I}$. The stabilizer variable can also serve to inform the error correction procedure. Every time we detect one or more stabilizer variables with a negative sign (this may happen after a stabilizer measurement), 
the decoder knows that at least one error affected the physical circuit. It proceeds to identifying the specific error and applies the corresponding correction. We define the error correction protocol as a function over stabilizer variables as follows:
\begin{definition}[Error correction protocol] Define \\
$\textbf{correct}(\svar_0,\svar_1,\cdots)$ as an error decoding and correction protocol by measuring $\svar_0,\svar_1,\cdots$.
\end{definition}

As an example, we present a snippet of {\langname} code that corresponds to  one stabilizer measurement of the surface code, where the parity qubit $s$ is connected to four data qubits $\{q_0,q_1,q_2,q_3\}$ with one Z-type stabilizer. Figure~\ref{fig:one-stabilizer}(a) shows the target stabilizer measurement circuit, while panel (b) shows an error correction program based on the stabilizer measurement in (a).  %
The basic idea of the program in Figure~\ref{fig:one-stabilizer}(b) is that, if the state of $q_0q_1q_2q_3$ is not stabilized by the stabilizer variable $\svar_0$, then some errors have happened and the QEC program should correct them.

Comparing to the \qwhilelang, {\langname} avoids the explicit introduction of parity qubits and the  implementation of stabilizer measurement circuits. Instead, it simply provides the stabilizer to measure as the value of variable $\svar$.
This approach makes {\langname} programs very flexible and independent from the specific implementation of stabilizer measurements. The latter would, for example, depend on architectural properties like the underlying hardware connectivity~\cite{Chamberland2020TopologicalAS}.

While being optimized for QEC codes, {\langname} contains all the program constructs necessary to describe general quantum programs. Actually, all programs in {\qwhilelang} can be translated into {\langname} by setting the stabilizer variable to a Pauli Z operator when qubits need to be measured.

\begin{figure}
    \centering
    \begin{subfigure}[b]{0.25\textwidth}
         \centering
         \includegraphics[width=\textwidth]{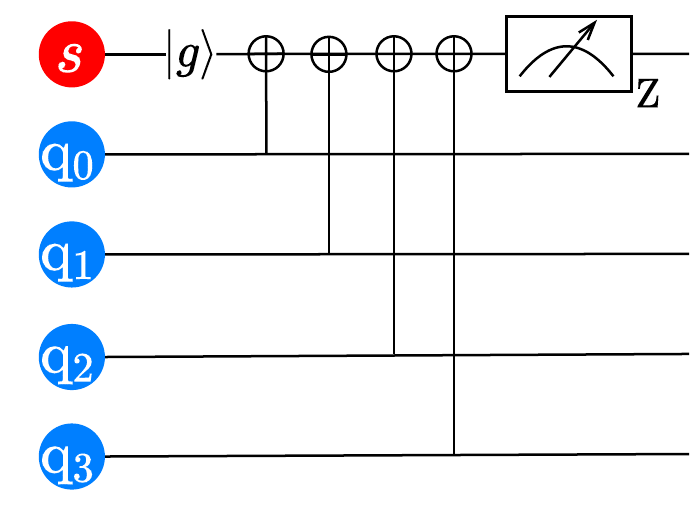}%
         \subcaption{}
    \end{subfigure}
    \begin{subfigure}[b]{0.20\textwidth}
    \begingroup \fontsize{\myfontsize}{\mylinesize}
    \addtolength{\jot}{-4pt}
    \begin{align*}
        & \svar_0 \coloneqq Z_{0}Z_{1}Z_{2}Z_{3} \\ 
        & \textbf{if}\ M[\svar_0, q_3q_2q_1q_0]\ \textbf{then} \\
        & \quad \textbf{skip} \\
        & \textbf{else} \\
        & \quad \text{// correct error....}\\
        & \quad \cdots \\
        & \quad \text{// Recover signal} \\
        & \quad \svar_0 \coloneqq -\svar_0\ 
    \end{align*}%
    \endgroup
    \subcaption{}
    \end{subfigure}
    \caption{An example for {\langname}. 
    (a) is the stabilizer measurement for $Z_{0}Z_{1}Z_{2}Z_{3}$. (b) is an error correction program associated with (a). The detailed error correction operation depends the error correction protocol and are omitted here.  
        }
    \label{fig:one-stabilizer}%
\end{figure}

\input{figtex/lang-operational}
\input{figtex/lang-denotational}

\subsection{Operational Semantics}\label{subsect:operational}

The \textit{operational semantics} of the proposed {\langname} are presented in Figure~\ref{fig:qec-lang-op}.
Different from the \qwhilelang, {\langname} denotes the state in QEC programs by the tuple $(\rho, \sigma)$, where $\rho$ is a partial density matrix that describes the current state of $\bar{q}$ , and $\sigma$ represents the current state of stabilizer variables. 
The quantum state $\rho$ can be regarded as a function over quantum variables $q$, and $\rho(q)$ represents the reduced partial density matrix where quantum variables except $q$ are all traced out.
Likewise, $\sigma$ represents a function over stabilizer variables, and $\sigma(\svar)$ is defined to be the current value of $\svar$.
The stabilizer state $\sigma$ is functionally similar to the classical program state~\cite{Winskel1993TheFS}, and 
we can define the substution rule for $\sigma$ in a similar way, which is then used in Figure~\ref{fig:qec-lang-op}.

\begin{definition}[Substitution rule] The substitution rules for stabilizer state $\sigma$ are defined as follows,
\begin{align*}
\sigma[s/\svar_i](\svar_j) &= \begin{cases}
s, \text{ if }\ i = j;\\
\sigma(\svar_j), \textbf{ otherwise} 
\end{cases}.
\end{align*}
\end{definition}

Rules in Figure~\ref{fig:qec-lang-op} are self-explained and represent reformulation of concepts familiar in quantum computing. The notation in these rules follows the convention in programming language research, for example, the expressions over the bar in the inference rules are \textit{premises} while the expression under the bar is \textit{conclusion}.
For pure quantum state operations in Figure~\ref{fig:qec-lang-op}, the operational semantics follows the flow in {\qwhilelang}~\cite{Ying2012FloydhoareLF}.  
For stabilizer related operations, we introduce extra operational semantics for the unary stabilizer expression $s_e^u$, as shown in the top right corner of Figure~\ref{fig:qec-lang-op}. We then process the assignment operation on stabilizer variables with the substitution rule.
When measuring the stabilizer variable $\svar$, its sign will get flipped if the current quantum state is not a +1 eigenstate of  $\svar$. Thus, we include one operation in the ``If -1'' and ``While -1'' rule to take care of the sign flipping on $\svar$.

To illustrate the use of the operational semantics in Figure~\ref{fig:qec-lang-op}, we revisit the program in Figure~\ref{fig:one-stabilizer}(b). We use $(0,\{\})$ to represent the initial state of $(\rho, \sigma)$ before the program.

\begin{example}[Error correction experiment]
Assume the initialization $q_3q_2q_1q_0 \coloneqq \ket{0000}$ on data qubits is distorted by noise and data qubits are assigned to be $\ket{0001}$. This may happen when, after the initialization in the logical subspace, a Pauli X error affects one of the qubits. 
For illustration purposes we consider that such error was on qubit $q_0$ and therefore use $q_0\coloneqq Xq_0;$ to correct the error. Notice that, in general, the error correcting protocol $\textbf{correct}(\svar_0, \cdots)$ depends on the outcome of multiple stabilizer measurements.
The program in Figure~\ref{fig:one-stabilizer}(b) then becomes
\begin{align*}
    \prog \equiv \ &q_3q_2q_1q_0 \coloneqq \ket{0001}; \svar_0 \coloneqq Z_{0123} ; \\
    &\textbf{if}\ M[\svar_0, q_3q_2q_1q_0]\ \textbf{then}\ \textbf{skip}\ \\
    &\textbf{else}\ q_0\coloneqq Xq_0; \svar_0 \coloneqq -\svar_0\ \textbf{end}.
\end{align*}
We write $Z_{0}Z_{1}Z_{2}Z_{3}$ as $Z_{0123}$ for simplicity.
Then the evaluation of $\prog$ with the operational semantics proceeds as follows:

\noindent
$
    \conf{\prog, \rho}  = \conf{q_3q_2q_1q_0 \coloneqq \ket{1110}; \svar_0 \coloneqq Z_{0123} ; \textbf{if}\ M[\svar_0, \\q_3q_2q_1q_0]\textbf{then}\ 
    \textbf{skip}\ \textbf{else}\ q_0\coloneqq Xq_0; \svar_0 \coloneqq -\svar_0\ \textbf{end}, (0, \{\})} \\
     \to \conf{\svar_0 \coloneqq Z_{0123} ; \textbf{if}\ M[\svar_0, q_3q_2q_1q_0]\ \textbf{then}\ \textbf{skip}\ 
    \  \textbf{else} \ q_0\coloneqq Xq_0; \svar_0 \coloneqq -\svar_0\ \textbf{end}, (\ket{0001}\bra{0001}, \{\})} \\
     \to \conf{\textbf{if}\ M[\svar_0, q_3q_2q_1q_0]\ \textbf{then}\
    \textbf{skip}\ \textbf{else}\ q_0\coloneqq Xq_0; \svar_0 \coloneqq -\svar_0\ \textbf{end}, (\ket{0001}\bra{0001}, \{\svar_0=Z_{0123}\})} \\
     \to \conf{q_0\coloneqq Xq_0; \svar_0 \coloneqq -\svar_0, (\ket{0001}\bra{0001}, \{\svar_0=-Z_{0123}\})} \\
     \to \conf{\svar_0 \coloneqq -\svar_0, (\ket{0000}\bra{0000}, \{\svar_0=-Z_{0123}\})} \\
     \to \conf{E, (\ket{0000}\bra{0000}, \{\svar_0=Z_{0123}\})}
$
\end{example}

\subsection{Denotational Semantics}\label{subsect:denotational}
The denotational semantics of the {\langname}  is given in Figure~\ref{fig:lang-denote}. The program $\prog$ is denoted as a super-operator $\denot{\prog}$ that acts on $(\rho, \sigma)$. While most rules in Figure~\ref{fig:lang-denote} are  self-explained, the \textbf{while} rule relies a partial order on $(\rho, \sigma)$ to compute the fixed point, i.e., the lowest upper bound ($\sqcup$) for the complete partial ordering (CPO) of  $\textbf{while}^{(k)}$.

To define the partial order on state $(\rho, \sigma)$, we first need to define a partial order on the state of stabilizers. Consider a trivial lattice on the stabilizer set $S$, where $s_i \sqsubseteq s_j$ if $s_i = I$, and $s_i$ and $s_j$ cannot be compared if they are both not identity. We can then define the partial order on $\sigma$ as follows: \\
\centerline{$\sigma_1 \sqsubseteq \sigma_2 \ \textbf{if}\ \sigma_1(\svar) \sqsubseteq \sigma_2(\svar), \forall \svar \in \textbf{sVar},$} \\
which immediately induces a partial order on $(\rho, \sigma)$: \\
\centerline{
$(\rho_1, \sigma_1) \sqsubseteq (\rho_2, \sigma_2) \ \textbf{if}\ \rho_1 \sqsubseteq \rho_2 \ \textbf{and}\ \sigma_1 \sqsubseteq \sigma_2,$} \\
where $\rho_1 \sqsubseteq \rho_2$ means that $Tr(O\rho_1) \le Tr(O\rho_2)$, for any semi-positive observable $O$. %

The reason to select the partial order above is that it strictly ensures the consistency of fault-tolerant computation. If any error happens in one while loop and does not get corrected in time, the stabilizer state $\sigma$ will get some variables flipped. Let the resulted stabilizer state be $\sigma'$. Obviously, there does not exist any $\sigma_1$ s.t. $\sigma \sqsubseteq \sigma_1$ and $\sigma' \sqsubseteq \sigma_1$. 
In this case, we just set $\sqcup_{k=0}^{\infty} \denot{\textbf{while}^{(k)}}(\rho, \sigma)$ to be $\perp$, the bottom element of the CPO~\cite{Winskel1993TheFS}, which does not provide any information for the program.  On the other hand, if all errors get corrected in each loop, the resulted state of \textbf{while} can be calculated coordinate-wisely: \\
\centerline{$(\sqcup_{k=0}^{\infty} \denot{\textbf{while}^{(k)}}\rho, \sqcup_{k=0}^{\infty} \denot{\textbf{while}^{(k)}}\sigma)$}.

We connect the denotational semantics to the operational semantics through the following proposition:

\begin{proposition}[Equivalence of the denotational semantics and the operational semantics]\label{prop:equiv}
For a strict QEC program $\prog$ that corrects errors when the errors appear, we have
$
    \denot{\prog}(\rho, \sigma) \equiv \sum \{  (\rho', \sigma'):\conf{\prog, (\rho, \sigma)} \to^* \conf{E, (\rho', \sigma')} \},
$
where $\to^*$ denotes the reflective, transitive closure of $\to$,  and $\{\cdot\}$ represents a multi-set.
\end{proposition}
\myproof{
Except the while loop, other statements can be proved trivially by structural induction. For the while loop, the consistency can be proved with the assumption of a just-in-time error correction.}

%% file: figtex/lang-operational.tex
\begin{figure*}
\begingroup \fontsize{\myfontsize}{\mylinesize}
\begin{tcolorbox}[colback=yellow!10!white,
                  colframe=white!20!black]
\begin{equation*}
\begin{aligned}[c]
    \text{(Skip)}& & &\infer[]{\conf{\mathbf{skip}, (\rho, \sigma)} \to \conf{E, (\rho, \sigma)}}{} \\
    \text{(Initialization)}& & &\infer[]{ \conf{q:=\qstate{0}, (\rho, \sigma)} \to \conf{E, (\rho_0^q, \sigma)}}{}\\
    \text{(Unitary)}& & &\infer[]{\conf{\bar{q}:=U[\overline{q}], (\rho, \sigma)}\to \conf{E, (U\rho U^\dagger, \sigma)}}{}\\
     \text{(Sequence E)}& & &\infer[]{\conf{E;\prog_2, (\rho, \sigma)}\to \conf{\prog_2, (\rho,\sigma)}}{}\\
\end{aligned}
\begin{aligned}[c]
    \text{(Stabilizer exp)} & & & \infer[]{\conf{c\cdot s, \sigma} \to c\cdot s}{}
    \quad \infer[]{\conf{\pm\svar, \sigma} \to \pm\sigma(\svar)}{}, \ c \in \{1,-1,i,-i\}\\
    \text{(Assignment)} & & &\infer[]{\conf{\svar\coloneqq s_e^u, (\rho, \sigma)} \to \conf{E, (\rho, \sigma[c\cdot  s/\svar])}}{\conf{s_e^u, \sigma}\to c\cdot s,\ \ } \\
    \text{(Sequence)} & & &\infer[]{\conf{\prog_1;\prog_2, (\rho,\sigma)} \to \conf{\prog_1';\prog_2,(\rho',\sigma')}}{\conf{\prog_1, (\rho,\sigma)} \to \conf{\prog_1', (\rho',\sigma')}}\\
\end{aligned}
\end{equation*}
\begin{align*}
    \text{(If -1)} & & &\infer[]{\conf{\textbf{if}\ M[\svar,\bar{q}]\ \textbf{then}\  \prog_1\ \textbf{else}\ \prog_0\ \textbf{end}, (\rho, \sigma)} \to \conf{\prog_0, (M_0\rho M_0^\dagger, \sigma[-\sigma(\svar)/\svar])} }{}, M_0 = \frac{I-\svar}{2} \\ %
    \text{(If 1)}& & &\infer[]{\conf{\textbf{if}\ M[\svar,\bar{q}]\ \textbf{then}\  \prog_1\ \textbf{else}\ \prog_0\ \textbf{end}, (\rho, \sigma)} \to \conf{\prog_1, (M_1\rho M_1^\dagger, \sigma)} }{}, M_1 = \frac{I+\svar}{2} \\
    \text{(While -1)}& & &\infer[]{\conf{\textbf{while}\ M[\svar,\bar{q}]\ \textbf{do}\ \prog_1 \ \textbf{done}, (\rho, \sigma)} \to \conf{E, (M_0\rho M_0^\dagger, \sigma[-\sigma(\svar)/\svar])} }{} \\
    \text{(While 1)}& & &\infer[]{\conf{\textbf{while}\ M[\svar,\bar{q}]\ \textbf{do}\ \prog_1\ \textbf{done}, (\rho, \sigma)} \to \conf{\prog_1;\textbf{while}\ M[\svar,\bar{q}]\ \textbf{do}\ \prog_1\ \textbf{done}, (M_1\rho M_1^\dagger, \sigma)}}{}%
\end{align*}
\end{tcolorbox}
\endgroup
    \caption{Operational semantics of \langname. Notations:  $\rho_0^q = \ket{0}_q\bra{0}\rho \ket{0}_q \bra{0} + \ket{0}_q\bra{1}\rho \ket{1}_q \bra{0}$.  $\sigma(\svar)$ means to evaluate $\svar$ in state $\sigma$, while $\sigma[s/\svar]$ means to replace the evaluation of $\svar$ with $s$ in the context $\sigma$. Recall that $E$ indicates the empty program. %
    }
    \label{fig:qec-lang-op}%
\end{figure*}

%% file: figtex/lang-denotational.tex
\begin{figure*}
\begingroup \fontsize{\myfontsize}{\mylinesize}
\begin{tcolorbox}[colback=yellow!10!white,
                  colframe=white!20!black]
\begin{equation*}
\begin{aligned}
    \denot{\textbf{skip}}(\rho, \sigma) & = (\rho, \sigma) \\
    \denot{q \coloneqq \qstate{0}}(\rho, \sigma) & = (\rho_0^q, \sigma) \\
    \denot{\bar{q} \coloneqq U[\bar{q}]}(\rho, \sigma) & = (U\rho U^\dagger, \sigma)
\end{aligned}
\hspace{3cm}
\begin{aligned}
\denot{s}\sigma & = s \qquad \denot{\pm\svar}\sigma = \pm\sigma(\svar)\\
\denot{\svar\coloneqq s_e^u}(\rho, \sigma) & = (\rho, \sigma[\denot{s_e^u}\sigma/\svar]) \\
\denot{\prog_1;\prog_2}(\rho, \sigma) & = \denot{\prog_2}(\denot{\prog_1}(\rho,\sigma))
\end{aligned}
\end{equation*}%
\begin{align*}
    \denot{\textbf{if}\ M[\svar,\bar{q}]\ \textbf{then}\  \prog_1\ \textbf{else}\ \prog_0\ \textbf{end}}(\rho, \sigma) & =  \denot{\prog_1}(M_1\rho M_1^\dagger, \sigma) +  \denot{\prog_0}(M_0\rho M_0^\dagger, \sigma[-\denot{\svar}\sigma/ \svar])\\ %
    \denot{\textbf{while}\ M[\bar{q}]\ \textbf{do}\ \prog_1\ \textbf{done}}(\rho, \sigma) & =  \sqcup_{k=0}^{\infty} \denot{\textbf{while}^{(k)}}(\rho, \sigma)
\end{align*}
\end{tcolorbox}
\endgroup
    \caption{Denotational semantics of \langname. Notations: $\textbf{while}^{k}$ is the k-th unrolling of \textbf{while}.}
    \label{fig:lang-denote}%
\end{figure*}

%% file: AssnLang.tex
\section{{\assnname}}
In this section, 
we first introduce an expressive assertion language \textbf{\assnname} and then derive a Hoare logic to verify QEC programs.
\subsection{Syntax of \assnname}

Stabilizer is a kind of Hermitian operator and can be used  as predicate for QEC programs. 
We observe that  the exponential computational overhead on Hermitian-based predicates may be circumvented by using stabilizers as predicates. In fact arithmetic operations (e.g. addition and multiplication) between stabilizers can be completed within a time polynomial in the number of qubits.
This observation is particularly important for QEC programs in which the majority of logical operations can be described with a few stabilizers and the corresponding predicate transformation can be framed as the multiplication of stabilizers.

However, as predicates, stabilizers are not universal.
There are infinitely many quantum states that are not eigenstates of any non-identity stabilizer, e.g., $\ket{\psi} = \frac{\sqrt{3}}{2}\ket{0} + \frac{1}{2}\ket{1}$. 
Such limitation will cause difficulty in the verification of QEC programs. For example, if we are given some state that is not the +1 eigenstate of any stabilizer (this is possible to happen in future universal fault-tolerant computation), we cannot find any predicate except $I$ to accommodate such state.%
One well-studied way in the quantum information community to address this problem is to use the Pauli expansion of quantum Observable (Hermitian matrices)~\cite{nielsen2002quantum,Wilde2013QuantumIT}: \nothmskip
\begin{lemma}[Pauli expansion]\label{lem:pauli-expansion}
The quantum observable $O$ of a $n$-qubit system can be expressed as a linear combination of Pauli strings:
   $ O = \sum_i w^i\sigma^i_n$,
where $\sigma_n^i \in \{I, X, Y, Z\}^{\otimes n}$ is a length-$n$ Pauli string, and $w_i \in \mathbb{R}$ is its coefficient.
\nothmskip
\end{lemma}

The Pauli expansion motivates the following proposition which provides a universal way to deal with arbitrary logical states in QEC programs: 
\nothmskip
\begin{proposition}\label{prop:universal-stabilizer}
$\forall\, \ket{\psi} \in \mathcal{H}_n$, there is a $\projector$ which is a sum of stabilizers (Pauli strings), that satisfies $\projector \ket{\psi} = \ket{\psi}$, and $\projector \ne I$.
\nothmskip
\end{proposition}
\myproof{
Note that $(\ket{\psi}\bra{\psi})\ket{\psi} = \ket{\psi}$, $\ket{\psi}\bra{\psi}$ is a non-identity observable and by Pauli expansion, it can be represented by a linear combination of Pauli strings, i.e., stabilizers.}

Inspired by Lemma~\ref{lem:pauli-expansion} and Proposition~\ref{prop:universal-stabilizer}, we introduce arithmetic expressions of stabilizers, and define the stabilizer expression $s_e$ as follows, 
\begin{align}
    s_e \Coloneqq s \vsep \lambda_0 s_{e0} + \lambda_1 s_{e1}, \ \lambda_0, \lambda_1 \in \mathbb{C}
\end{align}
where $s$ is a stabilizer. $s_e$ is different from the $s_e^u$ used for stabilizer variables which only consists of unary operations on stabilizers. However, by describing both the stabilizer variable and the predicate within the language of stabilizer, we can easily incorporate the information from stabilizer measurement into predicates.

By Proposition~\ref{prop:universal-stabilizer},
$s_e$ is universal as $\forall\ket{\psi}, \exists s_e\ s.t.\ s_e\ket{\psi} = \ket{\psi}$. For example, the state $\ket{\psi} = \frac{\sqrt{3}}{2}\ket{0} + \frac{1}{2}\ket{1}$ is a +1 eigenstate of  $s_e \coloneqq\,\frac{1}{2}Z + \frac{\sqrt{3}}{2}X$. %
We then formulate the assertion language \textbf{{\assnname}} on QEC programs as follows:
\begin{align}
    A \Coloneqq s_e \vsep A_0 \wedge A_1 \vsep A_0 \vee A_1 \vsep A_0 \Rightarrow A_1.
\end{align}

When $A \coloneqq s_e$, we say a QEC program state $(\rho, \sigma)$ satisfies an assertion $A$ if $A\rho = \rho$ (for $\rho=\ket{\psi}\bra{\psi}$, $A\rho = \rho \Leftrightarrow A\ket{\psi} = \ket{\psi}$) , and $s_e$ is commutable with all stabilizer variables in $\sigma$. We denote this relation by $(\rho,\sigma) \models A$. Requiring $s_e$ to be commutable with stabilizer variables in $\sigma$ is essential for developing the quantum Hoare logic in section~\ref{subsect:partial}.

The semantics of $A_0 \wedge A_1$ and other Boolean expressions can then be derived by structural induction:
\begin{itemize}
    \item $(\rho,\sigma) \models A_1 \wedge A_2$ iff $(\rho,\sigma) \models A_1$ and $(\rho,\sigma) \models A_2$;
    \item $(\rho,\sigma) \models A_1 \vee A_2$ iff $(\rho,\sigma) \models A_1$ or $(\rho,\sigma) \models A_2$;
    \item $(\rho,\sigma) \models (A_1 \Rightarrow A_2)$ iff $((\rho,\sigma) \models A_1) \Rightarrow ((\rho,\sigma) \models A_2)$.
\end{itemize}

If an assertion $A$ is satisfied by any program states $(\rho, \sigma)$, we simply denote such property as $\models A$.

The following lemma presents an important result for {\assnname} that will be frequently utilized in later sections.

\begin{restatable}[Implication rule]{lemma}{implictrule}
	\label{lem:implicitrule}
	For stabilizer expressions,
	\begin{enumerate}
		\item If $(\rho, \sigma) \models s_{e0}$ and $(\rho, \sigma) \models s_{e1}$, we have $(\rho, \sigma) \models s_{e0}s_{e1}$ and $(\rho, \sigma) \models \lambda_0 s_{e0}+\lambda_1 s_{e1}$, where  $\lambda_0+ \lambda_1=1$.
		\item Assume $s_{e0}$ is not singular. If $(\rho, \sigma) \models s_{e0}$ and $(\rho, \sigma) \models s_{e1}s_{e0}$, we have $(\rho, \sigma) \models s_{e1}$.
		\item Assume $(a s_{e0} + bs_{e1})\rho = \rho$, every stabilizer in $\sigma$ is commutable with $s_{e0}$ and $s_{e1}$, and $(\rho, \sigma) \models  s_{e2}$, then
		$(\rho, \sigma) \models a s_{e0} + bs_{e1}s_{e2}$.
	\end{enumerate}
\end{restatable}
\myproof{\postpone{\ref{app:assn}}.}

Rules from classical Boolean predicates can also be used for {\assnname}, such as the rules for disjunction and conjunction. We will use these rules directly without extra description. Especially, the identity operator $I$ represents \textbf{True} and the empty operator $0$ represents \textbf{False} in {\assnname}.

\subsection{Partial Correctness}
\label{subsect:partial}

\begin{figure}
\begingroup\fontsize{\myfontsize}{\mylinesize}
    $\{A\}\textbf{skip}\{A\}\ \hfill\text{(Skip)}$%
    
    $\{A[\ket{0}/q]\} q\coloneqq \ket{0} \{A\} \hfill \text{(Initialization)}$ %
    
    $\{A\} \bar{q}\coloneqq U\bar{q} \{UAU^\dagger\} \hfill\text{(Unitary)}$\\
    \hfill  $U$ is a unitary, but written in the sum of stabilizers.

    $\{A\} \svar\coloneqq \pm\svar \{A\} \hfill \text{(Assignment)}$%
    
    $\{A\} \svar\coloneqq s \{A\} \hfill \text{(Assignment)}$ \\
    \hfill  where $s$ is commutable with $A$, otherwise $\{A\} \svar\coloneqq s \{I\}$.

    $\infer{\{A\}\prog_1;\prog_2\{B\}}{\{A\}\prog_1\{C\}\quad \{C\}\prog_2\{B\}}\hfill \text{(Sequencing)}$%

    $\infer{\{\sum_{i=0}^1 A_i M_i\}\textbf{if}\,M[\svar, \bar{q}]\,\textbf{then}\, \prog_1\,\textbf{else}\, \prog_0\,\textbf{end}\{B\}}{\{A_1 \wedge \svar\}\prog_1\{B\}\quad \{A_0 \wedge -\svar \}\prog_0\{B\}} \hfill \text{(Condition)}$
    \\

    $\infer{\{\sum_{i=0}^1 A_i M_i\} \textbf{while}\, M[\svar, \bar{q}]\, \textbf{do}\, \prog_1\, \textbf{end}\,\{ A_0 \wedge -\svar \} }{\{ A_1 \wedge \svar \}\prog_1\{\sum_{i=0}^1 A_i M_i\} }\hfill \text{(While)}$\\
    
    $\infer{\{A\}\prog\{B\} }{\models (A \Rightarrow A')\quad \{A'\}\prog\{B'\}\quad \models (B' \Rightarrow B)} \hfill \text{(Consequence)}$\\
    
\endgroup
    \caption{Hoare rules for partial correctness assertions when $A \coloneqq s_e$.
    }
    \label{fig:pca}%
\end{figure}

A partial correctness assertion in \assnname\ has the form:
\hfill{} $\{A\} c \{B\}$, where $A, B \in $ \assnname, and $c\in$ \langname. We first present the Hoare logic for partial correctness assertions in which the precondition $A$ is a stabilizer expression $s_e$, as shown in Figure~\ref{fig:pca}. We will extend the Hoare logic to Boolean expressions like $A_1 \wedge A_2$ in Proposition~\ref{lem:bool-assn}. %

The proof rules in Figure~\ref{fig:pca} are syntax-directed and reduce proving a partial correctness assertion of a compound statement to proving partial correctness assertions of its sub-statements. We only explain some rules below, since most rules are self-explained.

In the initialization rule, $(\rho, \sigma)\models A[\ket{0}/q]$ means that $A\rho_0^q = \rho_0^q$ and $A$ commutes with all stabilizer variables in $\sigma$.  This can be seen as the quantum version substitution rule. 
A more useful case of the initialization rule is when all qubits are reset to $\ket{0}$, and for the $n$-qubit system, we have 
\begin{align}
    \{I\}q_{n-1}\cdots q_0\coloneqq \ket{0}^{\otimes n}\{Z_0\wedge Z_1\wedge \cdots \wedge Z_n\}.
\end{align}

In the unitary rule, we represent unitary matrices as the sum of stabilizers in order to utilize the cheap computational cost of stabilizer multiplication. 

The rules for condition and while loop resembles their classical counterparts except the state may be changed by the branching condition. A direct derivative of the Condition rule 
is to make $A_1 = A_0 = A$ as follows,
\begin{lemma}
$\dfrac{\{A \wedge \svar\}\prog_1\{B\}\quad \{A \wedge -\svar \}\prog_0\{B\}}{\{A\}\textbf{if}\,M[\svar, \bar{q}]\,\textbf{then}\, \prog_1\,\textbf{else}\, \prog_0\,\textbf{end}\{B\}}
$.\end{lemma}
\myproof{
Note that $M_0 + M_1 = I$.
}

Likewise, by letting $A_1 = A_0 = A$, we have%
\begin{lemma}
$\dfrac{\{ A \wedge \svar \}\prog_1\{A\} }{\{A\} \textbf{while}\, M[\svar, \bar{q}]\, \textbf{do}\, \prog_1\, \textbf{end}\,\{ A \wedge -\svar \} }$.
\end{lemma}

The consequence rule is a powerful tool for the verification of QEC programs since it can encode facts of QEC codes into partial correctness assertions. 
The following example demonstrates the usage of the proposed Hoare rules, including the consequence rule:
\begin{example}
Assume $\prog\Coloneqq \svar\coloneqq Z_1; \textbf{if}\ M[\svar, q_1]\ \textbf{then}\\\textbf{skip} \textbf{else}\,q_1\coloneqq Xq_1;q_0\coloneqq Xq_0 \,\textbf{end}$. \\We prove  $\{Z_0Z_1\}\prog\{Z_0\}$ as follows:\\
$\{Z_0Z_1\}\svar\coloneqq Z_1; \{Z_0Z_1\}$ \hfill{} (Assignment) \\
$(Z_0Z_1) M_1 = (Z_0Z_1)\frac{I+Z_1}{2} = Z_0\frac{I+Z_1}{2}$, 
$(Z_0Z_1) M_0 = -Z_0\frac{I-Z_1}{2}$ \\
$\{Z_0Z_1 \wedge Z_1\}\textbf{skip}\{Z_0Z_1 \wedge Z_1\}$ \hfill{}(Skip) \\
$\{Z_0Z_1 \wedge -Z_1\}q_1\coloneqq Xq_1 \{-Z_0Z_1 \wedge Z_1\}$ \hfill{} (Unitary) \\
$\{-Z_0Z_1 \wedge Z_1\} q_0\coloneqq Xq_0 \{Z_0Z_1 \wedge Z_1\}$ \hfill{} (Unitary) \\
$\{Z_0Z_1 \wedge -Z_1\}q_1\coloneqq Xq_1; q_0\coloneqq Xq_0 \{Z_0Z_1 \wedge Z_1\}$ \hfill{} (Sequencing)\\
$\{Z_0Z_1= (Z_0Z_1)M_0+ (Z_0Z_1)M_1\}\textbf{if}\ M[\svar, \bar{q}]\ \textbf{then}\  \textbf{skip}\,\\\textbf{else}\,q_1\coloneqq Xq_1; q_0\coloneqq Xq_0 \,\textbf{end}\{Z_0Z_1 \wedge Z_1\}$ \hfill{} (Condition)\\
Then, $\{Z_0Z_1\}\prog\{Z_0Z_1 \wedge Z_1\}$ \hfill{} (Sequencing)\\
$Z_0Z_1 \wedge Z_1 \Rightarrow Z_0$ \hfill (Implication) \\
With the consequence rule, we have 
$\{Z_0Z_1\}\prog\{Z_0\}$.
\end{example}

Now we extend the Hoare rules in Figure~\ref{fig:pca} to other Boolean assertions in \assnname. \nothmskip
\begin{restatable}[]{proposition}{boolassn}
\label{lem:bool-assn}
	We restate the Hoare rules for classical Boolean assertions as follows, \\
	if $\{A_0\}\prog\{B_0\} \wedge \{A_1\}\prog\{B_1\}$,  $\{A_0\wedge A_1\}\prog\{B_0\wedge B_1\}$; \\
	if $\{A_0\}\prog\{B_0\} \vee \{A_1\}\prog\{B_1\}$,  $\{A_0\vee A_1\}\prog\{B_0\vee B_1\}$; \\
	$\{I\}\prog\{I\}, \{0\}\prog\{B\}$, where $B$ is any assertion, and $0$ represents an empty set of program states. For example, if $s_{e_1}$ and $s_{e_2}$ anti-commutes, $s_{e_1} \wedge s_{e_2} = 0$.
\end{restatable}\nothmskip
\myproof{
\postpone{\ref{app:assn}}.}

We prove a lemma for error correction which is frequently used in verification sections later.
\nothmskip
\begin{restatable}[Decoding correctness]{proposition}{decodecorrect}
\label{prop:decodecorrect}
Assume an valid error decoding and correction protocol for $\textbf{correct}$ function. Let $S$ be the set of all active stabilizer measurements in error correction, and  define $ A_S = \wedge_{s_i \in S} s_i$, then \\
$\{I \}\textbf{correct}(\svar_0, \svar_1, \cdots) \{ A_S\}$,\\
$\{A \wedge A_S \}\textbf{correct}(\svar_0, \svar_1, \cdots) \{A \wedge A_S\}$,\\ where $\svar_0, \svar_1, \cdots$ enumerate all elements of $S$.
\end{restatable}
\myproof{
\postpone{\ref{app:assn}}.}

Finally, we prove the soundness of Hoare rules in Figure~\ref{fig:pca}. %
\begin{restatable}[Soundness]{theorem}{soundness}
\label{thm:soundness}
The proof system in Figure~\ref{fig:pca} is sound for the partial correctness assertions.
\end{restatable}\nothmskip
\myproof{
\postpone{\ref{app:assn}}.}

%% file: Evaluation.tex
\section{Theoretical Analysis}
In this section, we give a theoretical analysis of %
both the program size  and the computational complexity of our framework for implementing and verifying surface codes~\cite{Fowler2012SurfaceCT, Horsman2012SurfaceCQ}, respectively.

\subsection{Program Size}
We first compare the program size (i.e., the number of statements) when implementing the surface code~\cite{Fowler2012SurfaceCT, Horsman2012SurfaceCQ}, in the qWhile-Lang (i.e., the {\qwhilelang}) and the {\langname} (see row 2-3 of Table~\ref{tab:comp}).
In surface code, we consider two approaches to encode a logical qubit, the planar code and the double defect code (detailed implementation of these codes can be found in \cite{Fowler2012SurfaceCT, Horsman2012SurfaceCQ}).
For the distance-$d$ surface code, the planar version requires $O(d^2)$ data qubits, $O(d^2)$ parity qubits as well as $O(d^2)$ stabilizers. The double-defect version introduces an overhead of a factor 10 in all the three quantities.

As a code size estimation of {\langname}, we only need one statement per stabilizer measurement, whereas the qWhile-Lang requires at least eight gate operations
to describe the circuit measuring a stabilizer~\cite{Fowler2012SurfaceCT}. Thus, {\langname} provides $8\times$ program size compression for the surface code implementations.

\subsection{Verification Complexity of Clifford Gates} %

The defining property of Clifford operations is that, given a Clifford gate $G$ and a stabilizer $s$, $GsG^\dagger$ must also be a stabilizer, i.e. Clifford operations do not increase the number of stabilizers in the assertion. 

By framing both assertions and unitary operations in the language of stabilizers, {\myFrameworkName} can efficiently processes the verification of Clifford operations. The efficiency stems from the low cost of multiplying stabilizers, which is $O(d)$ because the length of the stabilizers for logical states is at most $d$ for a distance-$d$ surface code.
In this way we avoid representing stabilizers as exponentially-large matrices.
Therefore, \myFrameworkName~only incurs $O(d^3)$ computational overhead for the planar surface code and $O(10d^3)$ computational overhead for the double-defect surface code.

However, 
the vanilla quantum Hoare logic in qWhile-Lang can not exploit the property of Clifford operations and the low computational complexity of stabilizer multiplication. 
The Clifford operations are treated like any other unitary operations and the predicate in qWhile-Lang is a Hermitian matrix
of size $O(2^{n_d + n_p}\times 2^{n_d + n_p})$, where $n_d$ is the number of data qubits and $n_p$ is the number of parity qubits. Hoare rules with such predicates incurs at least $O(2^{n_d + n_p}\times 2^{n_d + n_p})$ computational overhead. Thus, the verification with qWhile-Lang requires $O(8d^24^{2d^2})$ time for the planar surface code, and $O(80d^24^{20d^2})$ for the double-defect surface code.
In summary, the proposed language design, assertion design, and the logic proof system can significantly simplify the verification of all Clifford operations of stabilizer codes.

\begin{table}[]
\resizebox{\columnwidth}{!}{
\begin{tabular}{|p{1.8cm}|p{2cm}|p{3cm}|p{2.5cm}|}
\hline
      Metric                &       Method                             & Planar surface code & Double-defect surface code \\ \hline
Statements \#           & qWhile-Lang & $O(8d^2)$           & $O(80d^2)$                 \\ \cline{2-4} 
                                         & {\myFrameworkName}               & $O(d^2)$            & $O(10d^2)$                 \\ \hline
Verification  & qWhile-Lang & $O(8d^24^{2d^2})$     & $O(80d^24^{20d^2})$             \\ \cline{2-4} 
                  Complexity                      & {\myFrameworkName}               & $O(d^3)$            & $O(10d^3)$                 \\ \hline             
\end{tabular}
}
\caption{
Comparison of {{\myFrameworkName}} and  qWhile-Lang on implementing and verifying surface codes. 
}
\label{tab:comp}
\end{table}

\subsection{Verification Complexity of T Gate}

The logical T gate is usually the most challenging problem in quantum program verification in general.
The T gate, loosely speaking, represents a fundamental boundary between classical and quantum computing. 
A quantum program with T gates cannot be efficiently and precisely simulated or verified on a classical computer. 
In \myFrameworkName, the number of stabilizers in our predicate will rapidly increase when the program to be verified has some T gates.

With this being said, we argue that verifying a QEC implementation of \textbf{one} logical T gate could be easier in many cases. The exact verification efficiency would be determined by the amount of the non-Clifford operations involved in the implementation of a logical T gate.
In the surface code, for example, there is only \textbf{one} non-Clifford single-qubit physical  gate~\cite{Fowler2015MinimumWP} for a logical T gate. 
The verification complexity will remain $O(d^2)$ because the number of stabilizer terms in a predicate is still $O(1)$. As such, \myFrameworkName~can still hold the exponential advantage for surface code. We remark again that such advantages come from our stabilizer-centric design in developing the verification framework.

\section{Case Study I: Repetition Code}
In this section and the next section, we give step-by-step case study on two well-known QEC codes to guide through the usage of our framework. 
For each QEC code, we first express its implementation in our {\langname}. %
Then we verify the correctness of the logical operation with our proof system and show that the implemented QEC code can correct local errors on the physical qubits.

We start from the quantum repetition code~\cite{nielsen2002quantum}, which is relatively simple with light error correction overhead. 
This code can correct bit-flip error or phase-flip error, but not when they happen simultaneously. 
The repetition code is mainly deployed on quantum architectures whose underlying physical qubits (e.g., the cat qubit built upon bosonic quantum system~\cite{Chamberland2020BuildingAF}) already have extremely low phase-flip (or bit-flip) error rates. %

\subsection{Quantum Repetition Code}
We consider a three-qubit quantum repetition code to simplify the discussion. At high level, the three-qubit code just encodes one logical qubit with three physical qubits. For example, the $\ket{000}$ state of three physical qubits represents the logical $\ket{0_L}$ state of a logical qubit, and the $\ket{111}$ state represents the logical $\ket{1_L}$ state. 
We first give the circuit diagrams for  primitive operations in quantum repetition code and their code implementations in {\langname} (Figure~\ref{fig:repcode}). 
The interested reader can find detailed explanations of the repetition code design in references~\cite{nielsen2002quantum, google50296}.

Figure~\ref{fig:repcode} (a) (b) (c) gives the implementations of the qubit initialization, the logical X gate $X_L = X_0X_1X_2$, and the logical Z gate $Z_L = Z_0Z_1Z_2$, respectively. 
The logical CNOT gate between two logical qubits can be implemented by imposing CNOT gates on three pairs of physical qubits, as shown in Figure~\ref{fig:repcode} (d).

\subsection{Verification of Logic Operations}

\input{figtex/repcode}

This part prove the correctness of the code segments in Figure~\ref{fig:repcode} with our frameworks. It contains two major steps, defining the predicates for each logical operation and constructing the proof. 

For the initialization operation, the expected behavior is that, for arbitrary input state, the output state should be in the logical state $\ket{0_L}$, which is the simultaneous eigenstate of the logical Z operator $Z_L$ and the stabilizers $Z_0Z_1$ and $Z_1Z_2$. Thus, we set the precondition to $\{I\}$ and the post-condition to $\{Z_L\wedge Z_0Z_1 \wedge Z_1Z_2\}$, as formulated below. \nothmskip
\begin{proposition}[Initialize to $\ket{0}$]
For the program $\prog$ in Figure~\ref{fig:repcode}(a), we have $\{I\}\prog\{Z_L\wedge Z_0Z_1 \wedge Z_1Z_2\}$.
\end{proposition}\nothmskip
\myproof{
For the initialization, $\{I\}q_0q_1q_2 \coloneqq \ket{000}\{Z_0 \wedge Z_1 \wedge Z_2\}$. And  $\{Z_0 \wedge Z_1 \wedge Z_2\}\svar_0\coloneqq Z_0Z_1\{Z_0\wedge Z_1 \wedge Z_2\}$ and $\{Z_0 \wedge Z_1 \wedge Z_2\}\svar_1\coloneqq Z_1Z_2\{Z_0\wedge Z_1 \wedge Z_2\}$. Since $Z_0\wedge Z_1 \wedge Z_2 \Rightarrow Z_0 Z_1 $ and $Z_0\wedge Z_1 \wedge Z_2 \Rightarrow Z_1 Z_2$, we have $Z_0\wedge Z_1 \wedge Z_2 \Rightarrow Z_0\wedge Z_1 \wedge Z_2 \wedge (Z_0 Z_1) \wedge (Z_1 Z_2)$. Then by Proposition~\ref{prop:decodecorrect},
$\{Z_0\wedge Z_1 \wedge Z_2 \wedge (Z_0 Z_1) \wedge (Z_1 Z_2)\}\textbf{correct}(\svar_0,\svar_1)\{Z_0\wedge Z_1 \wedge Z_2 \wedge (Z_0 Z_1) \wedge (Z_1 Z_2)\}$. By the consequence rule, we get $\{I\}\prog\{Z_L\wedge Z_0Z_1 \wedge Z_{1}Z_2\}$ since $Z_0\wedge Z_1 \wedge Z_2 \Rightarrow Z_L$.
}

We then verify the logical X operation.
It is sufficient to verify two cases, the output state $\ket{1_L}$ under the input state $\ket{0}_L$, and vice versa. Arbitrary logical states can be processed as the linear combination of these two cases by taking advantage of the linearity of the logical X operation.
Since $\ket{0}_L$ corresponds to the predicate $Z_L\wedge Z_0Z_1 \wedge Z_1Z_2$, and $\ket{1}_L$ corresponds to the predicate $-Z_L\wedge Z_0Z_1 \wedge Z_1Z_2$, we have the following proposition: \nothmskip
\begin{proposition}[Logical X gate]
For the program $\prog$ in Figure~\ref{fig:repcode}(b), we have $\{Z_L\wedge Z_0Z_1\wedge Z_1Z_2\}\prog\{-Z_L\wedge Z_0Z_1\wedge Z_1Z_2\}$ and $\{-Z_L\wedge Z_0Z_1\wedge Z_1Z_2\}\prog\{Z_L\wedge Z_0Z_1\wedge Z_1Z_2\}$.
\end{proposition} \nothmskip
\myproof{
Note that $X_0X_1X_2 Z_L X_0 X_1X_2 = -Z_1Z_2Z_3= - Z_L$, $X_0X_1\\X_2Z_0Z_1X_0X_1X_2 = Z_0Z_1$, $X_0X_1X_2Z_1Z_2X_0X_1X_2 = Z_1Z_2$.
}

Likewise, for thel logical Z gate, we only need to verify that, the precondition $\{X_L\wedge Z_0Z_1\wedge Z_1Z_2\}$ relates to the post-condition $\{-X_L\wedge Z_0Z_1 \wedge Z_1Z_2\}$, and vice versa. %
\nothmskip
\begin{proposition}[Logical Z gate]
For the program $\prog$ in Figure~\ref{fig:repcode}(c), $\{X_L\wedge Z_0Z_1\wedge Z_1Z_2\}\prog\{-X_L\wedge Z_0Z_1\wedge Z_1Z_2\}$ and $\{-X_L\wedge Z_0Z_1\wedge Z_1Z_2\}\prog\{X_L\wedge Z_0Z_1\wedge Z_1Z_2\}$.
\end{proposition} \nothmskip
\myproof{
Note that $Z_0Z_1Z_2X_LZ_0Z_1Z_2  = -X_0X_1X_2 = -X_L$.
}

The verification of the logical CNOT gate involves four preconditions: $Z_{L0}I_{L1}$, $X_{L0}I_{L1}$, $I_{L0}X_{L1}$, and $I_{L0}Z_{L1}$, where $Z_0Z_1 \wedge Z_1Z_2 \wedge Z_3Z_4 \wedge Z_4Z_5$ are omitted for simplicity. These four Pauli strings are able to represent any input state by multiplication and addition. The post-conditions for these four preconditions are $Z_{L0}I_{L1}$, $X_{L0}X_{L1}$, $I_{L0}X_{L1}$, and $Z_{L0}Z_{L1}$. While the first three post-conditions are straightforward to understand, we elaborate on the fourth post-condition. 
The precondition $I_{L0}Z_{L1}$ specifies pure states of the form $(a\ket{0}+b\ket{1})_{L0}\ket{0}_{L1}$. After the CNOT gate, the state becomes $a\ket{00}+b\ket{11}$ which is the +1 eigenstate of $Z_{L0}Z_{L1}$, for arbitrary $a$ and $b$. \nothmskip
\begin{restatable}[Logical CNOT]{proposition}{repcnot}
\label{prop:rep-cnot}
For the program $\prog$ in Figure~\ref{fig:repcode}(d), assume $A_S = Z_0Z_1\wedge Z_1Z_2\wedge Z_3Z_4\wedge Z_4Z_5$, we have\\
$\{Z_{L0} I_{L1}\wedge A_S\}\prog\{Z_{L0}I_{L1}\wedge A_S\}$, $\{X_{L0} I_{L1} \wedge A_S\}\prog\{X_{L0}X_{L1}\wedge A_S\}$, $\{I_{L0} X_{L1}\wedge A_S\}\prog\{I_{L0}X_{L1}\wedge A_S\}$, and $\{I_{L0} Z_{L1}\wedge A_S\}\prog\{\\Z_{L0}Z_{L1}\wedge A_S\}$.
\end{restatable} \nothmskip
\myproof{
Note that for control qubit $a$ and target qubit $b$, \\$\text{CNOT}_{ab} = \frac{1}{2}(I + X_b + Z_a - Z_aX_b)$. %
\postpone{\ref{app:rep}}.
}

\subsection{Verification on Noise Injection}
\myFrameworkName~can also reason about the correctness with hardware noise. We assume a minimum weight perfect matching error decoding~\cite{Fowler2015MinimumWP} and correction as follows, \nothmskip
\begin{program}[Quantum repetition code error correction]\label{prog:rep-decoder} For the quantum repetition code in Figure~\ref{fig:repcode}, define the error correction protocol as follows, for $\svar_0 = Z_0Z_1, \svar_1 = Z_1Z_2$, \\ $\textbf{correct}(\svar_0, \svar_1) \Coloneqq \\$
$
\qif{Z_0Z_1, q_0q_1}{\qif{Z_1Z_2, q_1q_2}{\textbf{skip}\\}{q_2\coloneqq Xq_2;\ \svar_1 \coloneqq -\svar_1}\\}{\qif{Z_1Z_2, q_1q_2}{q_0 \coloneqq Xq_0;\ \svar_0 \coloneqq -\svar_0 \\}{q_1\coloneqq Xq_1;\ \svar_0 \coloneqq -\svar_0; \ \svar_1 \coloneqq -\svar_1}}
$
\end{program} \nothmskip

In Figure~\ref{fig:noisyrepcode}(a), we present a noisy logical X gate where an X error occurs on $q_1$. We prove that the expected behavior of the noisy logical X gate is the same as that of the error-free logical X gate with the help of error correction.

\input{figtex/noisyrep}
\nothmskip
\begin{proposition}\label{prop:rep-noise-x} For the program $\prog$ in Figure~\ref{fig:noisyrepcode}(a), which implements a noisy logical X gate, \\
$\{Z_L \wedge Z_0Z_1 \wedge Z_1Z_2\}\prog\{-Z_L \wedge Z_0Z_1 \wedge Z_1Z_2\}$, \\
and $\{-Z_L \wedge Z_0Z_1 \wedge Z_1Z_2\}\prog\{Z_L \wedge Z_0Z_1 \wedge Z_1Z_2\}$.
\end{proposition} \nothmskip
\myproof{
We only prove $\{Z_L\wedge Z_0Z_1 \wedge Z_1Z_2\}{\prog}\{-Z_L\wedge Z_0Z_1 \wedge Z_1Z_2\}$ for simplicity.\\
$Z_L \wedge Z_0Z_1 \wedge Z_1Z_2 \Rightarrow Z_0 \wedge Z_1 \wedge Z_2$; \\
$\{Z_0\wedge Z_1 \wedge Z_2\}\svar_0 \coloneqq I \{Z_0\wedge Z_1 \wedge Z_2\}$; \\
$\{Z_0\wedge Z_1 \wedge Z_2\}\svar_1 \coloneqq I \{Z_0\wedge Z_1 \wedge Z_2\}$;
\\
$\{Z_0\wedge Z_1\wedge Z_2\}q_2q_1q_0 \coloneqq X_2X_1X_0q_2q_1q_0\{-Z_0\wedge -Z_1 \wedge -Z_2\}$; \\
$\{-Z_0\wedge -Z_1 \wedge -Z_2\}q_1 \coloneqq X_1q_1\{-Z_0\wedge Z_1 \wedge -Z_2\}$; 
\\
$\{-Z_0\wedge Z_1 \wedge -Z_2\}\svar_0 \coloneqq Z_0Z_1 \{-Z_0\wedge Z_1 \wedge -Z_2\}$;
\\
$\{-Z_0\wedge Z_1 \wedge -Z_2\}\svar_1 \coloneqq Z_1Z_2 \{-Z_0\wedge Z_1 \wedge -Z_2\}$;
\\
For the \textbf{correct} statement, \\$-Z_0\wedge Z_1 \wedge -Z_2\wedge Z_0Z_1 = 0$, $-Z_0\wedge Z_1 \wedge -Z_2 \wedge Z_1Z_2 = 0$, \\and $\{-Z_0\wedge Z_1\wedge -Z_2 \wedge -Z_0Z_1 \wedge -Z_1Z_2\}q_1\coloneqq Xq_1;\ \svar_0\coloneqq -\svar_0;\ \svar_1 \coloneqq -\svar_1\{-Z_0\wedge -Z_1 \wedge -Z_2 \wedge Z_0Z_1 \wedge Z_1Z_2\}$, \\
so $\{-Z_0\wedge Z_1 \wedge -Z_2\}\textbf{correct}(\svar_0, \svar_1)\{-Z_0\wedge -Z_1 \wedge -Z_2 \wedge Z_0Z_1 \wedge Z_1Z_2 \}$. Then by the consequence rule, we get $\{Z_L \wedge Z_0Z_1 \wedge Z_1Z_2\}\prog\{-Z_L \wedge Z_0Z_1 \wedge Z_1Z_2\}$.
}

However, the error correction protocol in Program~\ref{prog:rep-decoder} can not correct Z errors, as shown in the following proposition, \nothmskip
\begin{proposition} For the program $\prog$ in Figure~\ref{fig:noisyrepcode}(b), where a Z error happens on $q_1$, we have \\
$\{X_L \wedge Z_0Z_1 \wedge Z_1Z_2\}\prog\{-X_L \wedge Z_0Z_1 \wedge Z_1Z_2\}$ and $\{-X_L \wedge Z_0Z_1 \wedge Z_1Z_2\}\prog\{X_L \wedge Z_0Z_1 \wedge Z_1Z_2\}$, \\
which is not the desired behaviour of the logical X gate.
\end{proposition} \nothmskip
\myproof{
Similar to the proof in Proposition~\ref{prop:rep-noise-x}.
}

\section{Case Study II: Surface Code}

In this section, we present the verification of the double-defect surface code~\cite{Fowler2012SurfaceCT}. There are two types of stabilizers in surface code, of which one is called `Z-type' stabilizer as it only consists of Pauli Z operators and the other one is called `X-type' stabilizer as it only contains Pauli X operators. These two types of stabilizers together enable high error tolerance of surface code as well as  the ability to correct both bit-flip error and phase-flip error simultaneously. 
Implementing surface code has been pursued by several major quantum computing vendors including IBM~\cite{Chamberland2020TopologicalAS} and Google~\cite{ChenSatzingerAtalayaKorotkovDunsworthSankQui2}.

\subsection{Surface Code}

Double-defect surface code includes a series of logical operations to support fault-tolerant quantum computation, such as 
initialization, measurement, defect enlarging, defect shrinking, logical single-qubit gates (X, Z, H), qubit moving, braiding and the logical CNOT gate. 
In this section, we only verify qubit initialization, qubit moving, logical Pauli gates and logical qubit braiding. The verification of remaining operations is similar to or can be built upon the verified operations.
For example, qubit measurement is the reversing process of qubit initialization, 
and the logical CNOT gate is the concatenation of three braiding operations.

Without loss of generality, we only consider  the distance-3 surface code.
The logical operations shown in Figure~\ref{fig:surfop} is implemented in {\langname} in Figure~\ref{fig:surfcode}. %
We mainly focus on the operations on the X-cut qubit, which is a kind of logical qubit created by disabling X-type stabilizers.
We present the initialization operation in Figure~\ref{fig:surfop}(a) which initializes a X-cut qubit to the logical state $\ket{+_L}$, i.e., the +1 eigenstate of the logical X operator $X_L$ in Figure~\ref{fig:surfop}.
Figure~\ref{fig:surfop}(b) %
shows the logical Z gate $Z_L$. 
We then implements the qubit moving operation and the loigcal H gate shown in Figure~\ref{fig:surfop}(c)(d).  The qubit moving operation 
will not change the logical state. For the logical H gate, we only presents a simplified version which is enough to demonstrate the core idea of the logical H gate in ~\cite{Fowler2012SurfaceCT}. Finally, we include the verification of qubit initialization (to $\ket{0_L}$), the logical X gate  and the braiding operation in  Appendix~\ref{app:surf}. %

\textbf{All proofs in this section will be postponed to Appendix~\ref{app:surf}.}

\input{figtex/surffig}
\input{figtex/surfcode}

\subsection{Verification of Logic Operations}

As specified by the surface code~\cite{Fowler2012SurfaceCT}, any valid logical state should always be in the +1 eigenspace of all active stabilizers on the surface code array, $S = \{s_0,s_1,\cdots\}$. %
For simplification, we omit the stabilizers that does not involve in proof. For example, %
$(\ket{0_L}\bra{0_L},\sigma)\models Z_L \wedge s_0 \wedge s_1 \cdots$ will be denoted by $(\ket{0_L}\bra{0_L},\sigma)\models Z_L$.
In the cases where we need to stress other active stabilizers in the array, we will have $(\ket{0_L}\bra{0_L},\sigma)\models Z_L \wedge A_S$, where $A_S = \wedge_{s\in S}s$.

We first verify the initialization to  $\ket{+_L}$. 
The expected functionality of the initialization operation is to prepare arbitrary state into a desired state, as shown in the following proposition. The precondition of the partial correctness is just $I$ which allows any state. As for the post-condition,
notice that $\ket{+_L}$ is stabilized by the stabilizer $X_0X_1X_2X_4$ (i.e., $X_L$) and other active stabilizers in $S$. 
The verification of initialization to $\ket{0_L}$ is similar and postponed to Appendix~\ref{app:surf}.
\nothmskip
\begin{proposition}[Initialize $\ket{+}_L$]
For the program $\prog$ in Figure~\ref{fig:surfcode}(a) which initializes a X-cut logical qubit to $\ket{+}_L$, as shown in Figure~\ref{fig:surfop}(a), we have
$\{I\}\prog\{X_0X_1X_2X_4 \wedge A_S\}$.
\end{proposition} \nothmskip

The verification of the logical Z gate is similar to that in quantum repetition code and the precondition and post-condition 
can be derived in a similar way.  The verification of logical X gate is similar and postponed to Appendix~\ref{app:surf}.
\nothmskip
\begin{proposition}[Logical Z gate]
	For program $\prog$ in Figure~\ref{fig:surfcode}(b) which implements the logical Z gate in Figure~\ref{fig:surfop}(b), we have $\{X_L\}\prog\{-X_L\}$ and $\{-X_L\}\prog\{X_L\}$, \\
	where $X_L = X_0X_1X_2X_4$.
\end{proposition} \nothmskip

To reason about qubit moving, the key is to prove that the logical state is preserved. 
We first represent the current state of data qubits with the stabilizer language. The following lemma shows that there is a one-to-one mapping between the stabilizer expressions and the logical  quantum states. %
\nothmskip
\begin{restatable}{lemma}{statestabilizer}
\label{lemma:state-stabilizer}
For a X-cut qubit state $\ket{\psi}$, if $\ket{\psi} = \alpha \ket{0}_L + \beta \ket{1}_L$ ($\vert\alpha\vert^2 + \vert\beta\vert^2 = 1$), then there is a unique $(a Z_L + b X_L)$ s.t. $(a Z_L + b X_L)\ket{\psi} = \ket{\psi}$, and in this case $a = \frac{\alpha^2 - \beta^2}{\alpha^2 + \beta^2}$ and $b = \frac{2\alpha\beta}{\alpha^2 + \beta^2}$ .

\noindent\textbf{Conversely}, for a X-cut qubit state $\ket{\psi}$, if $(\frac{\alpha^2 - \beta^2}{\alpha^2 + \beta^2} Z_L + \frac{2\alpha\beta}{\alpha^2 + \beta^2} X_L)\\ *\ket{\psi} = \ket{\psi}$, and $\ket{\psi}$ is in the space spanned by $\{\ket{0_L}, \ket{1_L}\}$, then $\ket{\psi} = \alpha \ket{0}_L + \beta \ket{1}_L$, up to a global phase.
\end{restatable} \nothmskip
We then apply Lemma~\ref{lemma:state-stabilizer} to verify the vertical qubit moving operation in Figure~\ref{fig:surfop}(c). The verification of the horizontal qubit moving is similar. The following proposition confirms that the logical state is not changed since the precondition and the post-condition have equal coefficients w.r.t. the logical X and logical Z operators. 
\nothmskip
\begin{restatable}[Vertical qubit moving]{proposition}{surfvqmov}
For program $\prog$ in Figure~\ref{fig:surfcode}(c) which implements the qubit moving operation in Figure~\ref{fig:surfop}(c) %
, we have 
$\{a Z_L + b X_L\}\prog\{a Z'_L + b X'_L\}$, 
where $Z_L = Z_0Z_1Z_2$, $X_L = X_{2}X_{3}X_{4}X_{6}$, $Z'_L = Z_{0}Z_{1}Z_{2}Z_{6}$, $X'_L = X_{6}X_{8}X_{9}X_{10}$, $a,b \in \mathbb{C}$.
\end{restatable} \nothmskip

While the implementation of the logical H gate requires isolating the defects of a logical qubit, the core of logical H operation is to perform local H gates for the area outlined in Figure~\ref{fig:surfop}(d), which alone is a distance-$3$ planar surface code. We will only verify that part of program for simplicity. Other parts of the logical H operation can be verified on top of verified operations above. Like in the case in quantum repetition code, we only need to verify that the logical X operator is transformed into the logical Z operator and the logical Z operator is transformed into the logical X operator. 
\nothmskip
\begin{proposition}[Logical H gate]
	For the program $\prog$ in Figure~\ref{fig:surfcode}(d) which implements the simplified logical H gate in Figure~\ref{fig:surfop}(d) , we have
	$\{Z_L\}\prog\{X_L' \}$, $\{X_L\}\prog\{Z_L'\}$, 
	where $Z_L = Z_{1}Z_{6}Z_{11}$, $X_L = X_{5}X_{6}X_{7}$, $Z_L' = Z_{5}Z_{6}Z_{7}$, $X_L' = X_{1}X_{6}X_{11}$ for the distance-$3$ planar surface code  outlined in Figure~\ref{fig:surfop}(d).
\end{proposition} \nothmskip

\subsection{Verification on Noise Injection}
\input{figtex/noisysurf}

To reason the correctness when noise exists, we assume a minimal weight perfect matching (MWPM) decoder~\cite{Fowler2015MinimumWP} and error correction for the surface code array.
When the error only happens on one qubit, it can be easily detected by the decoder and be corrected, as indicated below:
\nothmskip
\begin{proposition}
For the program $\prog$ in Figure~\ref{fig:noisysurfcode}(a) which implements a noisy version of the logical Z gate in Figure~\ref{fig:surfop}(b), \\
$\{X_L \wedge A_S\}\prog\{-X_L \wedge A_S\}, \{-X_L \wedge A_S\}\prog\{X_L \wedge A_S\}$, \\
where $X_L = X_{0}X_{1}X_{2}X_{4}$.
\end{proposition} \nothmskip

When we increase the Z error location by one, the error correction protocol may fail.
\nothmskip
\begin{proposition}\label{prop:noisyfail}
For the program $\prog$ in Figure~\ref{fig:noisysurfcode}(b), \\
$\{X_L \wedge A_S\}\prog\{X_L \wedge A_S\}, \{-X_L \wedge A_S\}\prog\{-X_L \wedge A_S\}$,\\
which is not the desired behavior of the logical Z gate.
\end{proposition} \nothmskip

The proposition~\ref{prop:noisyfail} is expected because a distance-$d$ surface code cannot correct errors on more than $\lfloor\frac{d}{2} \rfloor$ qubits.
More complicated cases can be proved in a similar way with our verification framework.

%% file: figtex/repcode.tex
\begin{figure*}
\captionsetup[subfigure]{labelformat=empty}
\begin{subfigure}[b]{0.19\textwidth}
\hspace{-10pt}\includegraphics[height=51pt]{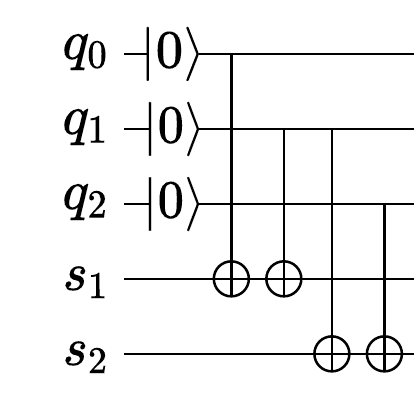}%

\begingroup \fontsize{\myfontsize}{\mylinesize}
$
q_2q_1q_0 \coloneqq \ket{000} \\
\svar_0\coloneqq Z_0Z_1 \\
\svar_1\coloneqq Z_1Z_2 \\
\textbf{correct}(\svar_0, \svar_1)
$
\endgroup
\subcaption{(a) Initialization.}
\end{subfigure} \tikz{\draw[-,black, densely dashed, thick](0,-1.05) -- (0,3.9);}\hspace{5pt}
\begin{subfigure}[b]{0.24\textwidth}
\hspace{0pt}\includegraphics[height=50pt]{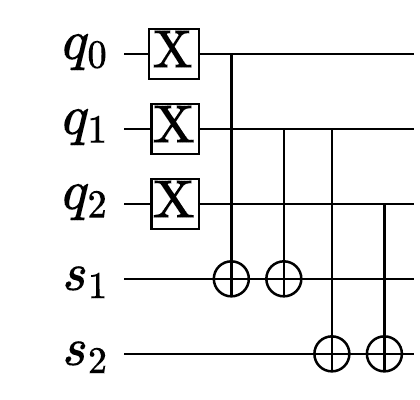}%

\begingroup \fontsize{\myfontsize}{\mylinesize}
$
\svar_0 \coloneqq I \\
\svar_1 \coloneqq I \\
q_2q_1q_0\coloneqq X_2X_1X_0q_2q_1q_0 \\ 
\svar_0 \coloneqq Z_0Z_1 \\
\svar_1 \coloneqq Z_1Z_2
$
\endgroup
\subcaption{(b) Logical X gate.}
\end{subfigure}\tikz{\draw[-,black, densely dashed, thick](0,-1.05) -- (0,3.9);}\hspace{5pt}
\begin{subfigure}[b]{0.24\textwidth}
\hspace{0pt}\includegraphics[height=50pt]{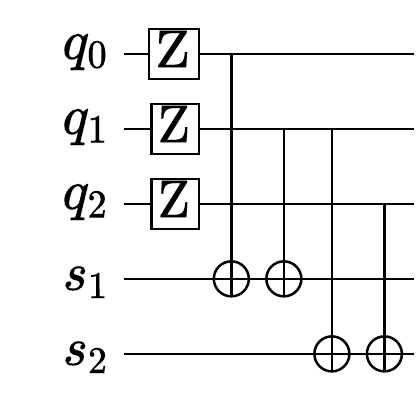}%

\begingroup \fontsize{\myfontsize}{\mylinesize}
$
\svar_0 \coloneqq I \\
\svar_1 \coloneqq I \\
q_2q_1q_0\coloneqq Z_2Z_1Z_0q_2q_1q_0\\ 
\svar_0 \coloneqq Z_0Z_1 \\
\svar_1 \coloneqq Z_1Z_2
$
\endgroup
\subcaption{(c) Logical Z gate.}
\end{subfigure}\tikz{\draw[-,black, densely dashed, thick](0,-1.05) -- (0,3.9);}\hspace{5pt}
\begin{subfigure}[b]{0.24\textwidth}
\hspace{0pt}\includegraphics[height=45pt]{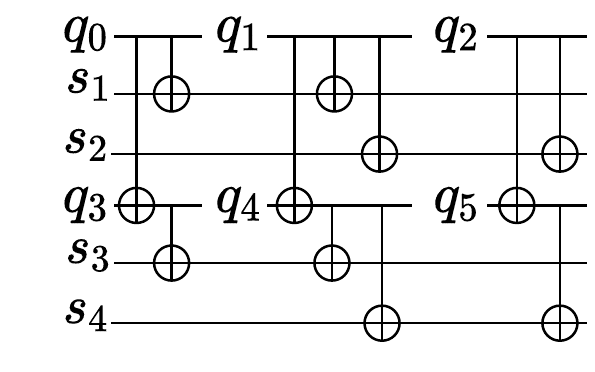}\\
\begingroup \fontsize{\myfontsize}{\mylinesize}
\addtolength{\jot}{-500pt}
$
\svar_0 \coloneqq I; \svar_1 \coloneqq I; \\
\svar_2 \coloneqq I; \svar_3 \coloneqq I; \\
q_0q_3 \coloneqq \text{CX}q_0q_3;\\
q_1q_4 \coloneqq \text{CX}q_1q_4; \\
q_2q_5 \coloneqq \text{CX}q_2q_5; \\
\svar_0 \coloneqq Z_0Z_1; \svar_1 \coloneqq Z_1Z_2 \\
\svar_2 \coloneqq Z_3Z_4; \svar_3 \coloneqq Z_4Z_5 
$
\endgroup
\subcaption{(d) Logical CX gate.}
\end{subfigure}
    \caption{Circuit diagrams for primitive operations in quantum repetition code and their code implementations in {\langname}. $s_1, s_2, s_3, s_4$ are parity qubits. $q_0, q_1, q_2, q_3, q_4, q_5$ are data qubits.
    }
    \label{fig:repcode}%
\end{figure*}

%% file: figtex/noisyrep.tex
\begin{figure}
\captionsetup[subfigure]{labelformat=empty}
\begin{subfigure}[b]{0.19\textwidth}
\begingroup \fontsize{\myfontsize}{\mylinesize}
$
\svar_0 \coloneqq I \\
\svar_1 \coloneqq I \\
q_2q_1q_0 \coloneqq X_2X_1X_0q_2q_1q_0 \\
\text{// an X error on }q_1; \\
q_1 \coloneqq X_1q_1 \\
\svar_0 \coloneqq Z_0Z_1\\
\svar_1 \coloneqq Z_1Z_2 \\
\textbf{correct}(\svar_0, \svar_1);
$
\endgroup
\subcaption{(a) An X error on $q_1$ }
\end{subfigure} \tikz{\draw[-,black, densely dashed, thick](-0.2,-1.05) -- (-0.2,2.75);}
\hspace{2pt}
\begin{subfigure}[b]{0.23\textwidth}
\begingroup \fontsize{\myfontsize}{\mylinesize}
$
\svar_0 \coloneqq I \\
\svar_1 \coloneqq I \\
q_2q_1q_0 \coloneqq X_2X_1X_0q_2q_1q_0 \\
\text{// an Z error on }q_1; \\
q_1 \coloneqq Z_1q_1; \\
\svar_0 \coloneqq Z_0Z_1\\
\svar_1 \coloneqq Z_1Z_2 \\
\textbf{correct}(\svar_0, \svar_1);
$
\endgroup
\subcaption{(b) A Z error on $q_1$}
\end{subfigure}
    \caption{Two noisy logical X gates.  }
    \label{fig:noisyrepcode}%
\end{figure}

%% file: figtex/surffig.tex
\begin{figure*}
\captionsetup[subfigure]{labelformat=empty}
\begin{subfigure}[b]{0.295\textwidth}
\includegraphics[width=\textwidth]{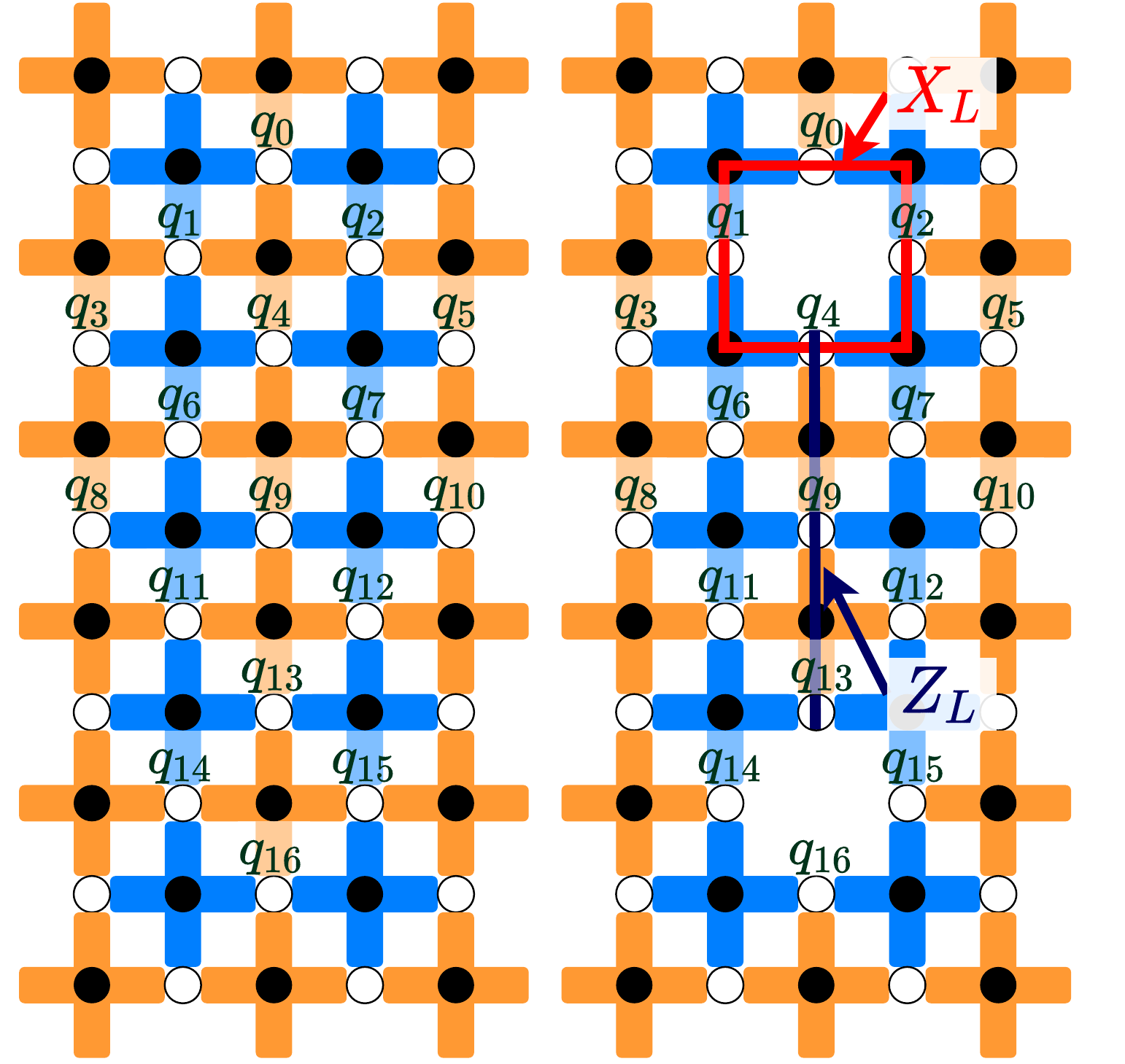}
\subcaption{(a) Initializes $\ket{+_L}$}
\end{subfigure}
\hspace{-7pt}
\tikz{\draw[-,black, densely dashed, thick](0,-1.05) -- (0,4.55);}
\begin{subfigure}[b]{0.152\textwidth}
\includegraphics[width=\textwidth]{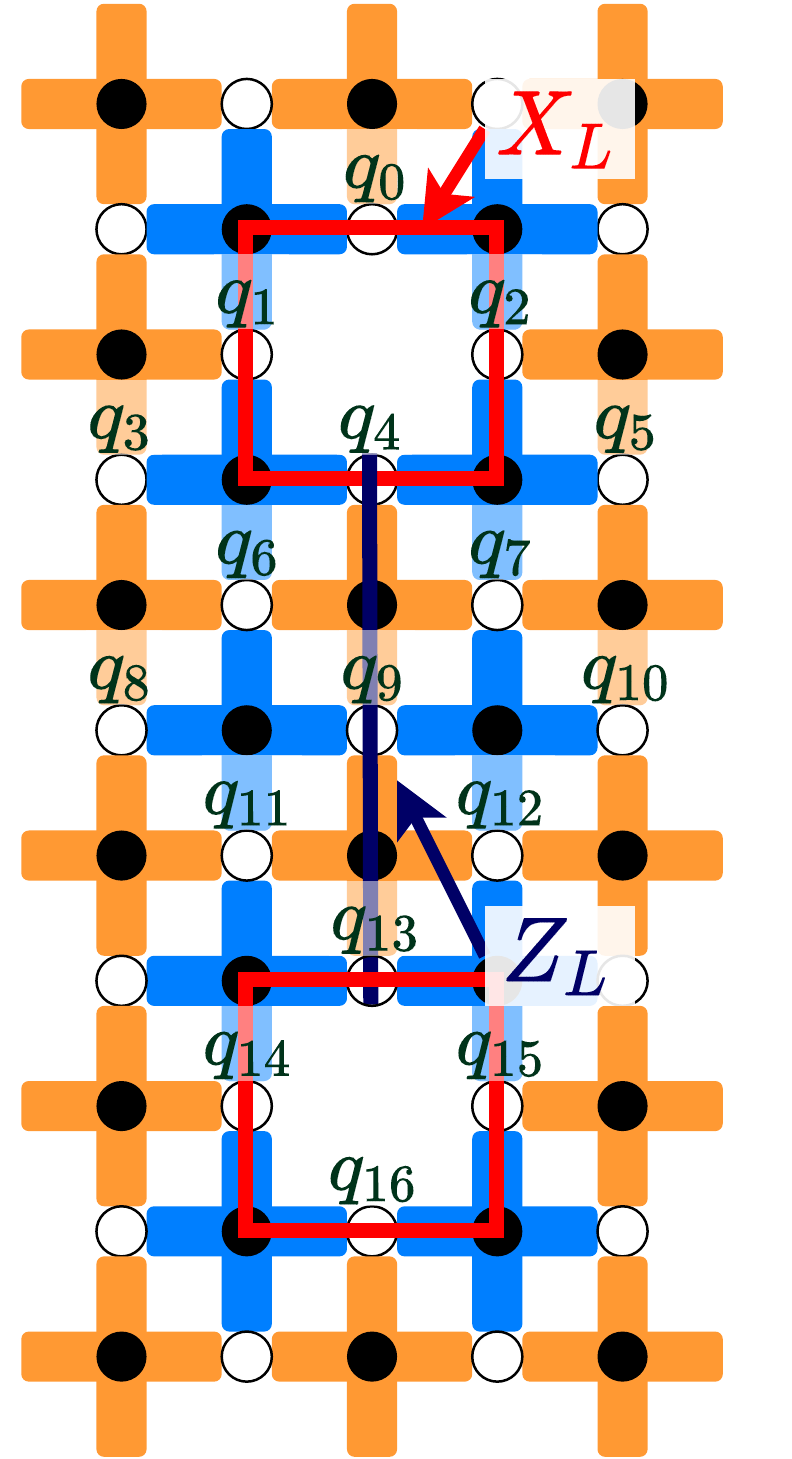}
\subcaption{\hspace{-2.5pt}(b) logical Z gate $Z_L$}
\end{subfigure}
\hspace{-5pt}
\tikz{\draw[-,black, densely dashed, thick](0,-1.05) -- (0,4.55);}
\begin{subfigure}[b]{0.365\textwidth}
\includegraphics[width=\textwidth]{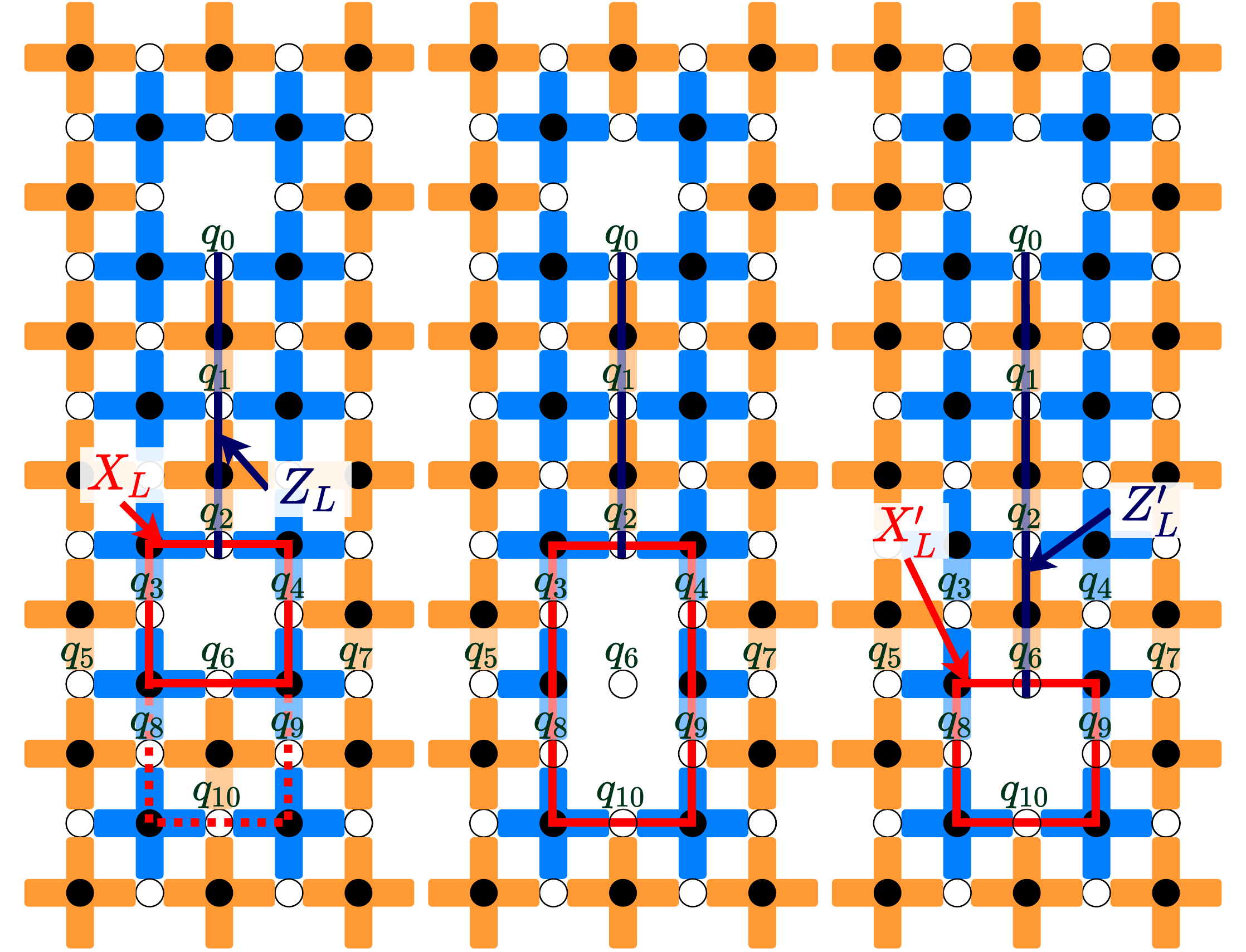}
\subcaption{(c) Vertical qubit moving}
\end{subfigure}
\tikz{\draw[-,black, densely dashed, thick](0,-1.05) -- (0,4.55);}
\hspace{-5pt}
\begin{subfigure}[b]{0.157\textwidth}
\includegraphics[width=\textwidth]{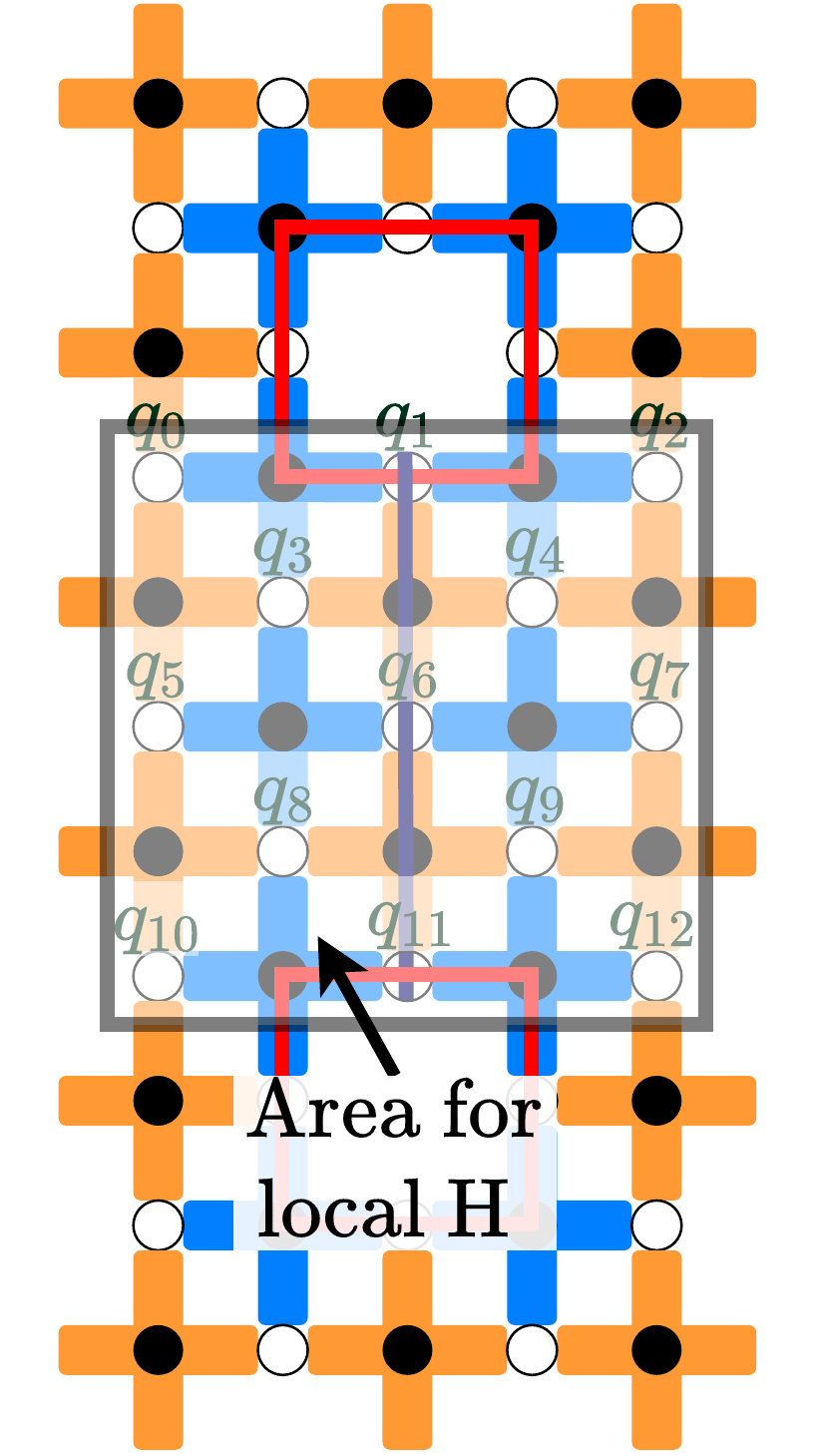}
\subcaption{(d) Logical H gate}
\end{subfigure}
    \caption{Primitive operations in the double-defect surface code . X stabilizers are yellow, and Z stabilizers are blue.}\label{fig:surfop}
\end{figure*}

%% file: figtex/surfcode.tex
\begin{figure*}
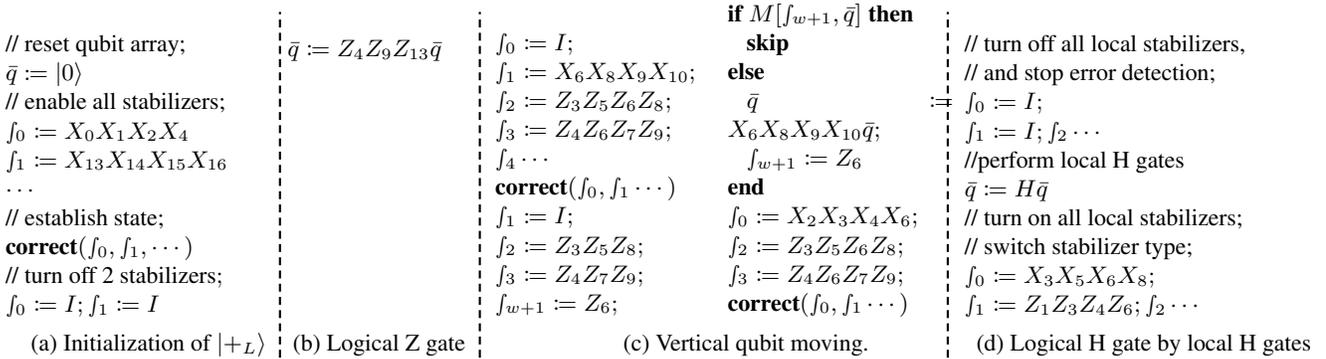

	\captionsetup[subfigure]{labelformat=empty}
	\begin{subfigure}[b]{0.20\textwidth}
		\begin{subfigure}[b]{\textwidth}
		\begingroup
		\fontsize{\myfontsize}{\mylinesize}
			$
			\text{// reset qubit array;} \\
			\bar{q} \coloneqq \ket{0}\\
			\text{// enable all stabilizers;} \\
			\svar_0 \coloneqq X_0X_1X_2X_4 \\
			\svar_1 \coloneqq X_{13}X_{14}X_{15}X_{16}\\
			\cdots \\
			\text{// establish state;} \\
			\textbf{correct}(\svar_0, \svar_1, \cdots) \\
			\text{// turn off 2 stabilizers}; \\
			\svar_0 \coloneqq I;\svar_1 \coloneqq I
			$
		\endgroup
			\subcaption{(a) Initialization of $\ket{+_L}$}
		\end{subfigure}
	\end{subfigure}
	\hspace{-10pt}
	\tikz{\draw[-,black, densely dashed, thick](0,-1.05) -- (0,3.55);}
	\begin{subfigure}[b]{0.15\textwidth}
		\begin{subfigure}[b]{\textwidth}
			\begingroup
			\fontsize{\myfontsize}{\mylinesize}
			$\bar{q}\coloneqq Z_{4}Z_{9}Z_{13} \bar{q}$
		\endgroup\\\\\\\\\\\\\\\\\\
			\subcaption{\hspace{-12pt}(b) Logical Z gate}
		\end{subfigure}
			\end{subfigure}
	\hspace{-14pt}
	\tikz{\draw[-,black, densely dashed, thick](0,-1.05) -- (0,3.55);}
	\hspace{0.5pt}
	\begin{subfigure}[b]{0.35\textwidth}
		\begin{subfigure}[b]{0.45\textwidth}
		\begingroup
			\fontsize{\myfontsize}{\mylinesize}
			$\svar_0 \coloneqq I; \\
			\svar_1 \coloneqq X_{6}X_{8}X_{9}X_{10}; \\
			\svar_2 \coloneqq Z_{3}Z_{5}Z_{6}Z_{8}; \\
			\svar_3 \coloneqq Z_{4}Z_{6}Z_{7}Z_{9}; \\
			\svar_4 \cdots \\
			\textbf{correct}(\svar_0,\svar_1\cdots) $\\
			$
			\svar_1 \coloneqq I; \\ 
			\svar_2 \coloneqq Z_{3}Z_{5}Z_{8}; \\
			\svar_3 \coloneqq Z_{4}Z_{7}Z_{9}; \\
			\svar_{\stabnum +1} \coloneqq Z_6; 
			$
		\endgroup
		\end{subfigure}
		\begin{subfigure}[b]{0.45\textwidth}
			\begingroup
			\fontsize{\myfontsize}{\mylinesize}
			$
			\qifn{\svar_{\stabnum +1},
				\bar{q}}{\textbf{skip}}{\bar{q}\coloneqq X_6X_{8}X_{9}X_{10}\bar{q};\\\myquad \svar_{\stabnum +1} \coloneqq Z_6} \\
			\svar_0 \coloneqq X_{2}X_{3}X_{4}X_{6}; \\
			\svar_2 \coloneqq Z_{3}Z_{5}Z_{6}Z_{8}; \\
			\svar_3 \coloneqq Z_{4}Z_{6}Z_{7}Z_{9}; \\
			\textbf{correct}(\svar_0,\svar_1\cdots)
			$
			\endgroup
		\end{subfigure}
		\subcaption{(c) Vertical qubit moving.}
	\end{subfigure}
	\hspace{-24pt}
	\tikz{\draw[-,black, densely dashed, thick](0,-1.05) -- (0,3.55);}
	\hspace{1pt}
	\begin{subfigure}[b]{0.25\textwidth}
		\begingroup
			\fontsize{\myfontsize}{\mylinesize}
		$
		\text{// turn off all local stabilizers,}\\ \text{// and stop error detection;}\\
		\svar_0 \coloneqq I; \\
		\svar_1 \coloneqq I; \svar_2\cdots \\
		\text{//perform local H gates} \\
		\bar{q} \coloneqq H\bar{q} \\
		\text{// turn on all local stabilizers;}\\\text{// switch stabilizer type;}\\
		\svar_0 \coloneqq X_3X_5X_6X_8; \\
		\svar_1 \coloneqq Z_1Z_3Z_4Z_6; \svar_2\cdots 
		$
		\endgroup
		\subcaption{(d) Logical H gate by local H gates}
	\end{subfigure}
	\caption{Implementation of primitive operations of the double-defect surface code in {\langname}. }
	\label{fig:surfcode}
\end{figure*}

%% file: figtex/noisysurf.tex
\begin{figure}
\captionsetup[subfigure]{labelformat=empty}
\begin{subfigure}[b]{0.23\textwidth}
$
q_{4}q_{9}q_{13} \coloneqq Z_{4}Z_{9}Z_{13} q_{4}q_{9}q_{13} \\
\text{// an Z error on qubit }q_9; \\
q_{9} \coloneqq Z_{9}q_{9} \\
\text{// assume } \svar_0,\svar_1,\cdots \\
\text{// be the active stabilizers;} \\
\textbf{correct}(\svar_0,\svar_1,\cdots)
$
\subcaption{(a) a Z error on $q_9$.}
\end{subfigure}
\hspace{-8pt}
\tikz{\draw[-,black, densely dashed, thick](-0.2,-1.05) -- (-0.2,1.90);}
\hspace{2pt}
\begin{subfigure}[b]{0.23\textwidth}
$
q_{4}q_{9}q_{13} \coloneqq Z_{4}Z_{9}Z_{13} q_{4}q_{9}q_{13} \\
\text{// Z errors on qubit } q_9, q_{13} \\
q_9 \coloneqq Z_9q_9; \\
q_{13} \coloneqq Z_{13}q_{13}; \\
\textbf{correct}(\svar_0,\svar_1,\cdots)
\\
$
\subcaption{(b) two Z errors on $q_9$, $q_{13}$.}
\end{subfigure}
    \caption{Noisy programs of the logical Z gate  in  Figure~\ref{fig:surfop}(b).  }
    \label{fig:noisysurfcode}%
\end{figure}

%% file: conclusion.tex
\section{Conclusion}

Quantum error correction is the bedrock of fault-tolerant quantum computation, and its verification 
is of significant importance for the forthcoming large-scale quantum computing.
In this paper, we propose {\myFrameworkName}, an efficient verification framework for stabilizer codes. 
 {\myFrameworkName} first comes with a concise language, {\langname}, which incorporates  stabilizers to represent QEC programs.
Stabilizers together with stabilizer expressions are also used as predicates in our new assertion language, {\assnname}. 
We then derive a sound quantum Hoare logic to efficiently reason about the correctness of QEC programs.
Finally, We evaluate  {\myFrameworkName} by a theoretical complexity analysis and case studies on two QEC codes.
We believe this work will spark more interest in the verification of QEC programs which may become a prevalent programming paradigm in the near future.

%% file: appendix.tex
\onecolumn
\section{Appendix}

\subsection{Proof in {\assnname}}
\label{app:assn}

\implictrule*
\begin{proof} 
	For the first rule, note that  $(s_{e0}s_{e1})\rho = s_{e0}(s_{e1}\rho ) = s_{e0}\rho  = \rho$. 
	Also, $\forall \svar \in \sigma$, $s_{e0}s_{e1} \svar = s_{e0}\svar s_{e1} = \svar s_{e0}s_{e1}$. Thus, $(\rho, \sigma) \models s_{e0}s_{e1}$. $(\rho, \sigma) \models \lambda_0 s_{e0} + \lambda_1 s_{e1} $ can be proved similarly. \\
	For the second rule, note that $s_{e1}\rho = s_{e1}(s_{e0}\rho ) =
	(s_{e1}s_{e0})\rho = \rho$, and $\forall \svar \in \sigma$, $(\svar s_{e_1})s_{e_0} = s_{e_1}s_{e_0}\svar = (s_{e_1}\svar) s_{e_0}$. Since $s_{e_0}$ is not singular, we have $\svar s_{e_1}= s_{e_1}\svar$, thus $(\rho, \sigma) \models s_{e1}$. \\
	For the final rule, notice that $(a s_{e0} + bs_{e1}s_{e2})\rho = a s_{e0}\rho + bs_{e1}s_{e2}\rho = a s_{e0}\rho + bs_{e1}\rho = \rho$.\\
	Finally, it is easy to see in all these rules, the stabilizer in $\sigma$ is commutable with the target stabilizer expressions. 
\end{proof}

\boolassn*
\begin{proof}
	We first prove the conjunction rule. Since $A_0\wedge A_1 \Rightarrow A_0$, $A_0\wedge A_1 \Rightarrow A_1$, then by the consequence rule, we have $\{A_0\wedge A_1\}\prog\{B_0\}$ and $\{A_0\wedge A_1\}\prog\{B_1\}$, i.e., $\{A_0\wedge A_1\}\prog\{B_0\wedge B_1\}$. For the disjunction rule, notice that if $(\rho,\sigma)\models (A_0\vee A_1)$, then either $(\rho,\sigma)\models A_0$ or $(\rho,\sigma)\models A_1$. 
	Finally,
	$\{I\}\prog\{I\}$ always holds since any state $(\rho,\sigma)$ satisfies $I$. $\{0\}\prog\{B\}$ is true because $(\rho,\sigma)\models 0 \Rightarrow \denot{P}(\rho,\sigma)\models B$.
\end{proof}

\decodecorrect*
\begin{proof}
First, any valid correction function will project the state into one quiescent state of the QEC code. It's the definition of QEC code error correction. \\
Second, note that the assertion $A \wedge A_S$  represents error-free states in the QEC code, thus any valid \textbf{correct} protocol will place a \textbf{skip} statement for correcting the error-free state. 
Assume the \textbf{correct} protocol is implemented based on the look-up table, 
since $A \wedge A_S \wedge -\svar_i = 0$,
then by the condition rule and $\{0\}\prog\{A \wedge A_S\}$ (Lemma~\ref{lem:bool-assn}), we directly get $\{A \wedge A_S\}\textbf{correct}(\svar_0, \svar_1, \cdots) \{A \wedge A_S\}$.
\end{proof}

\soundness*
\begin{proof}
\setcounter{cnt}{0}
(\showcnt) Skip. Note than the skip rule does not change the program state.\\
(\showcnt) Initialization.
By the definition of the substitution rule, $(\rho, \sigma) \models A[\ket{0}/\rho]$ is equivalent to $(\rho_0^q, \sigma) \models A$, then the state after initialization $(\rho',\sigma) = (\rho_0^q,\sigma)$ also satisfies $A$. \\
(\showcnt) Unitary. Note that $(UAU^\dagger)(U\rho U^\dagger) = U A \rho U^\dagger$, so  \\$(UAU^\dagger)(U\rho U^\dagger) = (U\rho U^\dagger) \Leftrightarrow A\rho = \rho$.
\\
(\showcnt) Assignment. For the first rule, assume $(\rho, \sigma) \models A$, then $A$ is commutable with $\svar$. Then, $A$ is also commutable with $-\svar$. Thus, $(\rho,\sigma') = (\rho,\sigma[-\svar/\svar])$ also satisfies $A$. \\
The second rule is obviously correct, but it limits the selection of $A$. \\
(\showcnt) Sequencing. Assume $(\rho, \sigma) \models A$, then $\denot{P_0}(\rho, \sigma)\models C$ by the hypothesis $\{A\}P_0\{C\}$. On the other hand $\denot{P_0;P_1}(\rho,\sigma) = \denot{P_1}(\denot{P_0}(\rho, \sigma)) \models B$ by the hypothesis $\{C\}P_1\{B\}$.
\\
(\showcnt) Condition. 
First, $\sum A_i M_i$ is a legal stabilizer expression because $M_1 = \frac{I+\svar}{2}$ and $M_0 = \frac{I-\svar}{2}$ are legal stabilizer expressions. 
Assume $(\rho, \sigma) \models A$, then $\sigma(\svar)$ is commutable with $A$, so is $M_1$ and $M_0$. Thus, $A M_1\rho M_1^\dagger = M_1 A\rho M_1^\dagger = M_1\rho M_1^\dagger$. Likewise, we have $A M_0\rho M_0^\dagger = M_0\rho M_0^\dagger$. Let $A = \sum_i A_i M_i$, then $A M_1\rho M_1^\dagger = A_1M_1(M_1\rho M_1^\dagger) + A_0M_0(M_1\rho M_1^\dagger) = A_1(M_1\rho M_1^\dagger)$ since $M_1M_1 = M_1$, $M_1M_0 = 0$. Thus, we have $A_1M_1\rho M_1^\dagger = M_1\rho M_1^\dagger$. 
Since $\svar$ is commutable with both $A_1$ and $A_0$, we have $(M_1\rho M_1^\dagger,\sigma) \models A_1$ and $(M_0\rho M_0^\dagger,\sigma[-\svar/\svar]) \models A_0$. Also, $(M_1\rho M_1^\dagger,\sigma) \models \svar$ and $(M_0\rho M_0^\dagger,\sigma[-\svar/\svar]) \models -\svar$. Thus, if $(\rho,\sigma)\models \sum_i A_iM_i$, we have $(M_1\rho M_1^\dagger,\sigma) \models A_1\wedge \svar$ and $(M_0\rho M_0^\dagger,\sigma) \models A_0\wedge -\svar$.
Since $\{A_1 \wedge \svar\}P_1\{B\}$ and $\{A_0 \wedge -\svar\}P_0\{B\}$, by the semantics of the condition statement, we have $\{\sum A_i M_i\}\textbf{if}\,M[\svar, \bar{q}]\,\textbf{then}\, P_0\,\textbf{else}\, P_1\,\textbf{end}\{B\}$. \\
(\showcnt) While. The proof of the While rule is quite similar to that of the Condition rule. $\sum A_iM_i$ is called the invariant of the loop. If the execution enters the loop body, then by $\{ A_1 \wedge \svar \}P_0\{\sum A_i M_i\}$, we still have $(\rho, \sigma) \models \sum A_i M_i$ for the next loop. So, when the while loop terminates, we always have $(\rho, \sigma) \models  A_0 \wedge -\svar$. \\
To prove the While rule more formally, we only need to show the partial correctness holds for $\textbf{while}^{(k)}$, as $\textbf{while}$ is the disjunction of $\textbf{while}^{(k)}$, $k=0,1,2,\cdots$.
\\
(\showcnt) Consequence. Assume $(\rho,\sigma) \models A$, then $ (\rho,\sigma) \models A'$ by $\{A\Rightarrow A'\}$. Since $\{A'\}\prog\{B'\}$, we have $\denot{P}(\rho,\sigma) \models B'$. Then $\denot{P}(\rho,\sigma) \models B$ by $B'\Rightarrow B$. Thus, $\{A\}\prog\{B\}$.
\end{proof}

\subsection{Verification of Quantum Repetition Code}
\label{app:rep}
\repcnot*
\begin{proof}
First, for control qubit $a$ and target qubit $b$, $\text{CNOT}_{ab} = \frac{1}{2}(I + X_b + Z_a - Z_aX_b)$. Then\\
(1) $\{Z_{L0} I_{L1}\}\prog\{Z_{L0}I_{L1}\}$. Note that both $\text{CNOT}_{03}$, $\text{CNOT}_{14}$ and $\text{CNOT}_{25}$ are commutable with $Z_{L0}$, so $\text{CNOT}_{03}Z_{L0}\text{CNOT}_{03} = Z_{L0}\text{CNOT}_{03}\text{CNOT}_{03} = Z_{L0}$, , $\text{CNOT}_{14}Z_{L0}\text{CNOT}_{14}= Z_{L0}$ and $\text{CNOT}_{25}Z_{L0}\text{CNOT}_{25}= Z_{L0}$. \\
(2) $\{X_{L0} I_{L1}\}\prog\{X_{L0}X_{L1}\}$. Note that $\text{CNOT }_{03}X_{L0}\text{CNOT}_{03} = X_{L0}X_3$. Since $X_3$ is commutable with $\text{CNOT}_{14}$, $\text{CNOT}_{14}X_{L0} X_3 \text{CNOT}_{14} = (\text{CNOT }_{14}X_{L0}\text{CNOT}_{14})X_3 = X_{L0}X_4X_3$. Finally, $\text{CNOT}_{25}X_{L0}X_4X_3\text{CNOT}_{25} = X_{L0}X_5X_4X_3 = X_{L0}X_{L1}$. \\
(3) $\{I_{L0} X_{L1}\}\prog\{I_{L0}X_{L1}\}$. Note that both $\text{CNOT}_{03}$, $\text{CNOT}_{14}$ and $\text{CNOT}_{25}$ are commutable with $X_{L1}$. \\
(4) $\{I_{L0} Z_{L1}\}\prog\{Z_{L0}Z_{L1}\}$. Note that  $\text{CNOT}_{03}Z_{L1}\text{CNOT}_{03} = Z_{0}Z_{L1}$, $\text{CNOT}_{14}Z_{0}Z_{L1}\text{CNOT}_{14} = Z_{0}\text{CNOT}_{14}Z_{L1}\text{CNOT}_{14} = Z_{0}Z_{1}Z_{L1}$, and $\text{CNOT}_{25}Z_{0}Z_{1}Z_{L1}\text{CNOT}_{25} = Z_{0}Z_{1}\text{CNOT}_{25}Z_{L1}\text{CNOT}_{25} = Z_{0}Z_{1}Z_{2}Z_{L1} = Z_{L0}Z_{L1}$. \\
Finally, We can prove that $\{Z_0Z_1\}\prog\{Z_0Z_1\}$, $\{Z_1Z_2\}\prog\{Z_1Z_2\}$, $\{Z_3Z_4\}\prog\{Z_0Z_1Z_3Z_4\}$, $\{Z_4Z_5\}\prog\{Z_1Z_2Z_4Z_5\}$ in a similar way. Combing all these facts, we can prove the desired partial correctness on the logical CNOT gate.
\end{proof}

\subsection{Verification of the Surface Code}
\label{app:surf}

\begin{program}[Initialize $\ket{0_L}$]
For the initialization operation in the  figure below, which initializes an X-cut logical qubit to $\ket{0_L}$,\\
\includegraphics[width=0.5\textwidth]{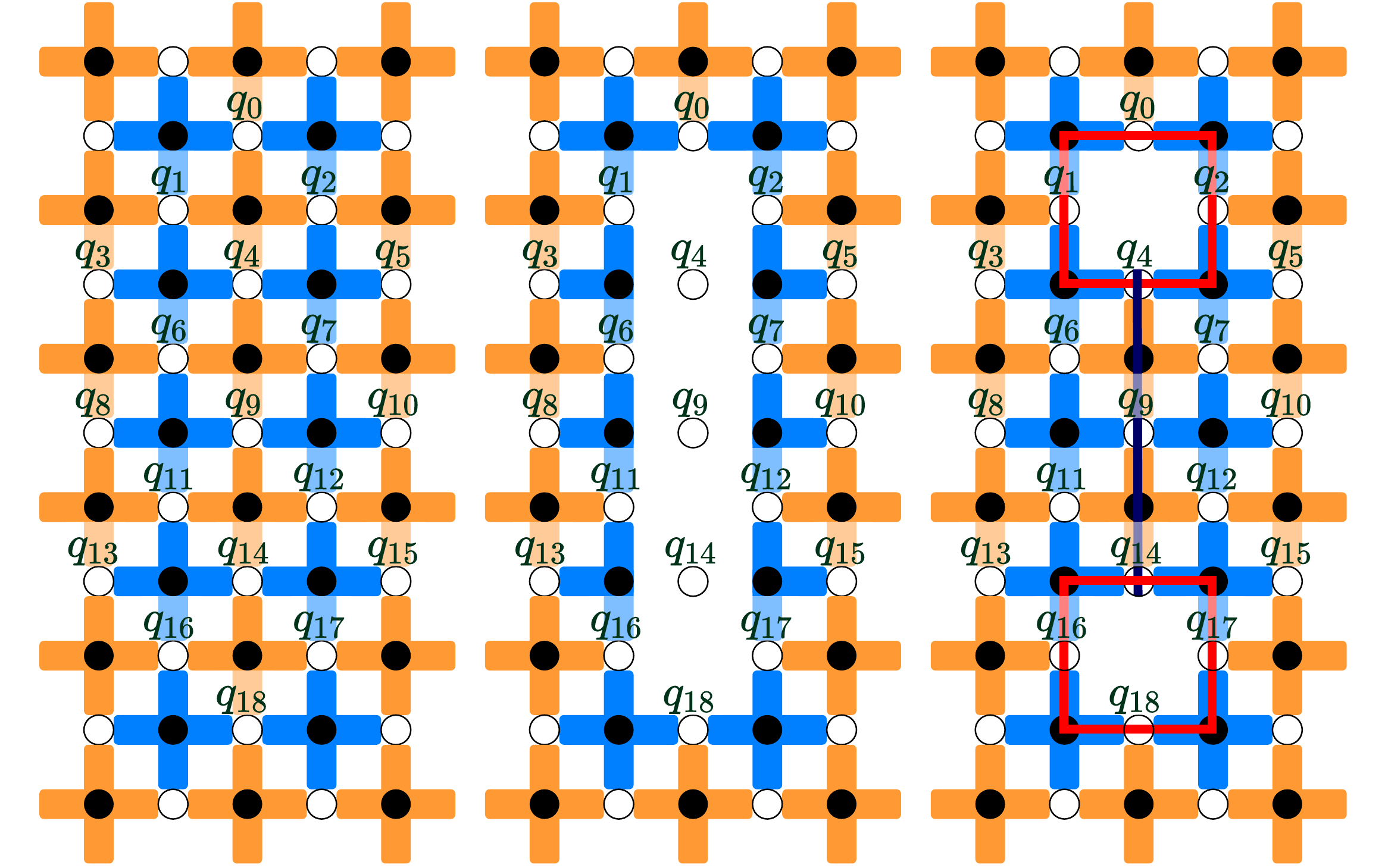}\\
\label{prog:initzero}
we have $\prog \Coloneqq
\bar{q} \coloneqq \ket{0};
\svar_0 \coloneqq X_0X_1X_2X_4;
\svar_1 \coloneqq X_4X_6X_7X_9;
\svar_2 \coloneqq X_9X_{11}X_{12}X_{14};
\svar_3 \coloneqq X_{14}X_{16}X_{17}X_{18}; \\
\svar_4 \coloneqq Z_1Z_3Z_4Z_6;
\svar_5 \coloneqq Z_2Z_4Z_5Z_7;
\svar_6 \coloneqq Z_{6}Z_{8}Z_{9}Z_{11};
\svar_7 \coloneqq Z_{11}Z_{13}Z_{14}Z_{16};
\svar_8 \coloneqq Z_{7}Z_{9}Z_{10}Z_{12};
\svar_{9} \coloneqq Z_{12}Z_{14}Z_{15}Z_{17};
\svar_{10} \coloneqq \cdots \\
\textbf{correct}(\svar_0,\svar_1,\cdots);
\svar_0 \coloneqq I; 
\svar_1 \coloneqq I;
\svar_2 \coloneqq I; 
\svar_3 \coloneqq I;
\svar_4 \coloneqq Z_1Z_3Z_6; 
\svar_5 \coloneqq Z_2Z_5Z_7; 
\svar_6 \coloneqq Z_{6}Z_{8}Z_{11}; 
\svar_7 \coloneqq Z_{11}Z_{13}Z_{16}; \\
\svar_8 \coloneqq Z_{7}Z_{10}Z_{12}; 
\svar_{9} \coloneqq Z_{12}Z_{15}Z_{17}; 
\svar_{10} \coloneqq \cdots \\
\svar_{\stabnum +1} \coloneqq Z_{4}; 
\svar_{\stabnum +2} \coloneqq Z_{9};
\svar_{\stabnum +3} \coloneqq Z_{14}; \\
\text{// set }q_4, q_9, q_{14}\text{ to }\ket{0}; \\
\qif{\svar_{\stabnum +1}, q_4 }{\textbf{skip}}{\bar{q} \coloneqq X_4X_6X_7X_9\bar{q}; \svar_{\stabnum +1} \coloneqq Z_{4}} %
\\
\textbf{if}\ M[\svar_{\stabnum +2}, q_9]\ \textbf{then}
\myquad \textbf{skip}\ 
\textbf{else}
\myquad \bar{q} \coloneqq X_9X_{11}X_{12}X_{14}\bar{q}; 
\svar_{\stabnum +2} \coloneqq Z_{9}\ 
\textbf{end} \\
\qif{\svar_{\stabnum +3}, q_{14}}{\textbf{skip}}{\bar{q} \coloneqq X_{14}X_{16}X_{17}X_{18}\bar{q}; \svar_{\stabnum +1} \coloneqq Z_{14}}; \\
\svar_1 \coloneqq X_4X_6X_7X_9;
\svar_2 \coloneqq X_9X_{11}X_{12}X_{14}$;
$\textbf{correct}(\svar_0,\svar_1,\cdots)$.
\end{program}

\begin{restatable}[Initialize $\ket{0_L}$]{proposition}{surfinitzerol}
For the program $\prog$ in Program~\ref{prog:initzero} which initializes a X-cut logical qubit to $\ket{0_L}$, \\
$\{I\}\prog\{Z_4Z_9Z_{14}\}$. Here $Z_4Z_9Z_{14}$ is the logical Z operator $Z_{L}$.
\end{restatable}
\begin{proof}
By Proposition~\ref{prop:decodecorrect}, after \textbf{correct} function, $(\rho,\sigma)\models (\svar_0 \wedge \svar_1 \wedge \svar_2 \cdots)$. The following stabilizer assignments which turn off X-stabilizers will just forward the precondition. \\
For simplicity, assume there are {\stabnum} stabilizers in the surface code array. let $\Lambda = \{0,\cdots,w-1\}$, then $(\svar_0 \wedge \svar_1 \wedge \svar_2 \cdots) = \wedge_{i\in \Lambda}\svar_i$. Since $\svar_0 \wedge \svar_1 \Rightarrow X_0X_1X_2X_6X_7X_9$ and $X_0X_1X_2X_6X_7X_9$ is commutable with $Z_4$, $\{\wedge_{i\in \Lambda}\svar_i\}\svar_{\stabnum +1} \coloneqq Z_{4}\{(\wedge_{i\in \Lambda\setminus \{0,1\}}\svar_i) \wedge X_0X_1X_2X_6X_7X_9\}$. Likewise, we know that after $\svar_{\stabnum +2} \coloneqq Z_{9}$, the precondition will become \\
$\{(\wedge_{i\in \Lambda\setminus \{0,1,2,3\}}\svar_i) \wedge X_0X_1X_2X_6X_7X_{11}X_{12}X_{16}X_{17}X_{18}\}$. \\
Note that $(\wedge_{i\in \Lambda\setminus \{0,1,2,3\}}\svar_i) \wedge X_0X_1X_2X_6X_7X_{11}X_{12}X_{16}X_{17}X_{18} \Rightarrow X_0X_1X_2X_6X_7X_{11}X_{12}X_{16}X_{17}X_{18} $. \\
Let $A = X_0X_1X_2X_6X_7X_{11}X_{12}X_{16}X_{17}X_{18}$, \\
$c = \qif{\svar_{\stabnum +1}, q_4}{\textbf{skip}}{\bar{q} \coloneqq X_4X_6X_7X_9\bar{q}; \svar_{\stabnum +1} \coloneqq Z_{4}}$.
It's easy to see that $\{A \wedge Z_4\} \textbf{skip} \{A \wedge Z_4\}$, and $\{A \wedge -Z_4\} q_4 \coloneqq Xq_4; \svar_{\stabnum +1}\coloneqq Z_4 \{A \wedge Z_4\}$. Thus, $\{A\}c\{A \wedge Z_4\}$. Then, after reset $q_{14}$ to $\ket{0}$, the precondition will become: \\
$\{A \wedge Z_4 \wedge Z_9 \wedge Z_{14}\}$.
Again, the following stabilizer assignments will just forward the precondition. By the implication rule, we have that $A \wedge Z_4 \wedge Z_9 \wedge Z_{14} \Rightarrow A \wedge Z_4Z_9Z_{14}$. Since $Z_4Z_9Z_{14}$ and all assertions in $A$ are commutable with stabilizers $\svar_0, \svar_1, \cdots$, we have $\{ A \wedge Z_4Z_9Z_{14}\} \textbf{correct}(\svar_0,\svar_1,\cdots) \{\wedge_{i\in \Lambda} \svar_i \wedge Z_4Z_9Z_{14}  \wedge X_0X_1X_2X_6X_7X_{11}X_{12}X_{16}X_{17}X_{18} \}$. Then by applying the consequence rule, we get $\{I\}\prog\{Z_4Z_9Z_{14}\}$.
\end{proof}

\begin{program}[Logical X gate] For the logical X gate $X_L$ in the Figure below:\\
\includegraphics[width=0.2\textwidth]{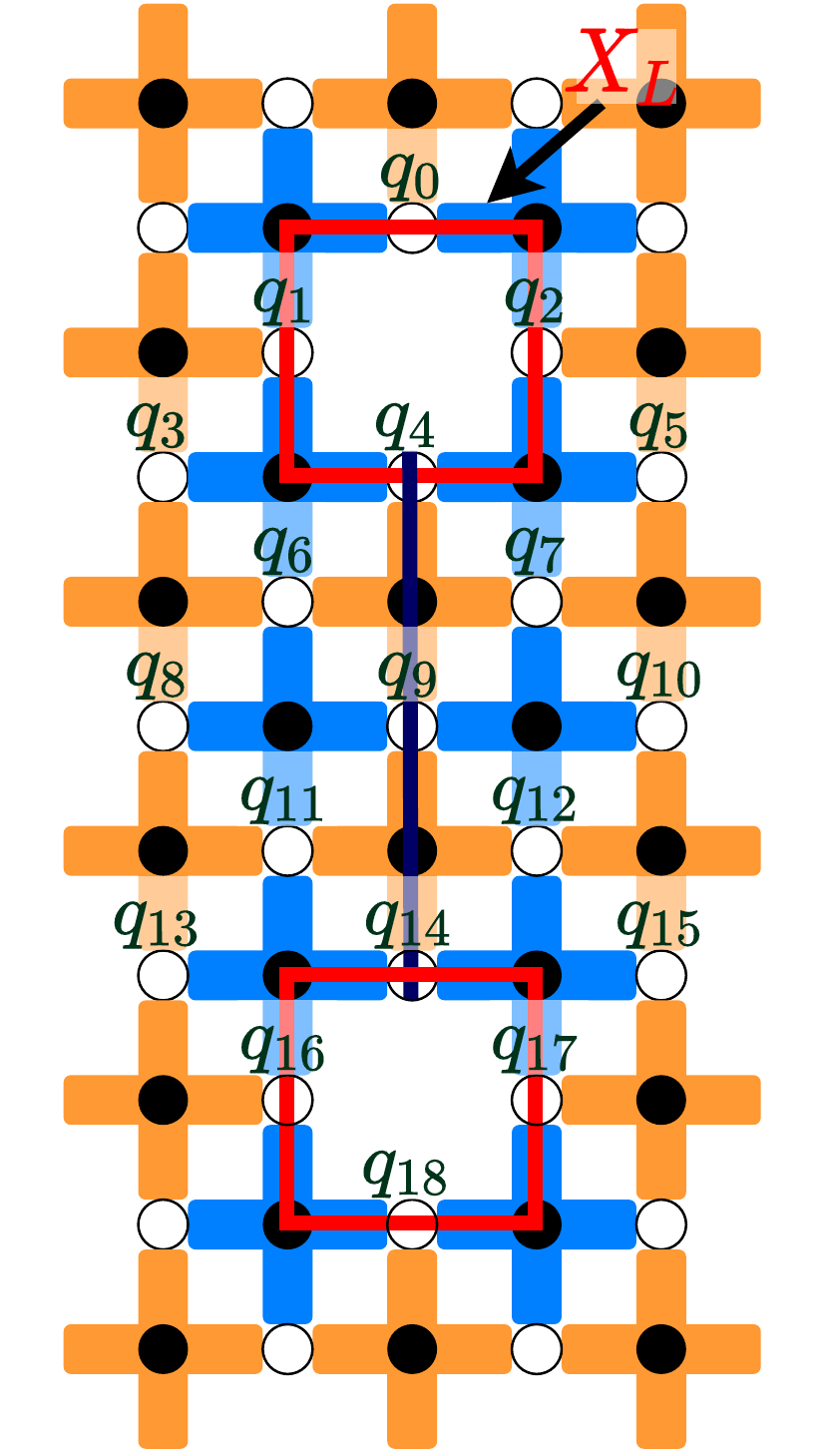} \\
we have $\prog \Coloneqq q_0q_1q_2q_4 \coloneqq X_{0}X_{1}X_{2}X_{4}q_0q_1q_2q_4$.
\end{program}

\begin{proposition}[Logical X gate]
For program $\prog$ in Figure~\ref{fig:surfcode}(d), we have  $\{Z_L\}\prog\{-Z_L\}$ and $\{-Z_L\}\prog\{Z_L\}$, \\
where $Z_L = Z_4Z_9Z_{14}$.
\end{proposition}
\begin{proof}
	Notice that $(X_L)(Z_L)(X_L)^\dagger = -Z_L$.
\end{proof}

\statestabilizer*
\begin{proof}
For the first part, we can get $a = \frac{\alpha^2 - \beta^2}{\alpha^2 + \beta^2}$ and $b = \frac{2\alpha\beta}{\alpha^2 + \beta^2}$ simply by solving the equation $(a Z_L + b X_L)\ket{\psi} = \ket{\psi}$. \\
For the second part, assume $\ket{\psi_0} = \alpha_0 Z_L + \beta_0 X_L$ and $\ket{\psi_1} = \alpha_1 Z_L + \beta_1 X_L$, if there is a $aZ_L + bX_L$ s.t. $(aZ_L + bX_L)\ket{\psi_0} = \ket{\psi_0}$ and $(aZ_L + bX_L)\ket{\psi_1} = \ket{\psi_1}$. Then, we have $a = (\frac{\alpha_0^2 - \beta_0^2}{\alpha_0^2 + \beta_0^2} = (\frac{\alpha_1^2 - \beta_1^2}{\alpha_1^2 + \beta_1^2}$, which is equivalent to $1 - \frac{2}{1 + (\frac{\alpha_0}{\beta_0})^2} = 1 - \frac{2}{1 + (\frac{\alpha_1}{\beta_1})^2}$. Thus, $(\frac{\alpha_0}{\beta_0})^2 = (\frac{\alpha_1}{\beta_1})^2$. On the other hand, $b = \frac{2\alpha_0\beta_0}{\alpha_0^2 + \beta_0^2} = \frac{2\alpha_1\beta_1}{\alpha_1^2 + \beta_1^2}$, which is equivalent to $\frac{\frac{\alpha_0}{\beta_0}}{1 + (\frac{\alpha_0}{\beta_0})^2} = \frac{\frac{\alpha_1}{\beta_1}}{1 + (\frac{\alpha_1}{\beta_1})^2}$. Thus, $\frac{\alpha_0}{\beta_0} = \frac{\alpha_1}{\beta_1}$, i.e., $\ket{\psi_0} = \ket{\psi_1}$ up to a global phase.
\end{proof}
\surfvqmov*
\begin{proof}
After the first \textbf{correction} function, the precondition is transformed into: $(a Z_L + b X_L)\wedge_i \svar_i$.
The three following stabilizer assignments will forward the precondition. Then by the implication rule, $(a Z_L + b X_L)\wedge_i \svar_i \Rightarrow (aZ_L + bX_{2}X_{3}X_{4}X_{8}X_{9}X_{10})\wedge_{i \ne 1}\svar_i$. so for the next stabilizer assignment $\svar_{\stabnum +1} = Z_6$, precondition $(aZ_L + bX_{2}X_{3}X_{4}X_{8}X_{9}X_{10}) \wedge_{i \ne 1}\svar_i$ will be forwarded. Note that $(aZ_L + b X_{2} X_{3} X_{4} X_{8} X_{9}X_{10}) \wedge_{i \ne 1}\svar_i \Rightarrow aZ_L + b X_{2} X_{3} X_{4} X_{8} X_{9}X_{10}$,
let $A = aZ_L + b X_{2} X_{3} X_{4} X_{8} X_{9}X_{10}$, \\
$c = \qif{\svar_{\stabnum +1}, \bar{q}}{\textbf{skip}}{\bar{q}\coloneqq X_6X_{8}X_{9}X_{10}\bar{q}; \svar_{\stabnum +1}=Z_6}$. \\
For the if statement, $\{A \wedge \svar_{w+1}\}\textbf{skip}\{A \wedge \svar_{w+1}\}$ and $\{A \wedge -\svar_{w+1}\}\bar{q}\coloneqq X_6X_{8}X_{9}X_{10}\bar{q}; \svar_{\stabnum +1}=Z_6\{A \wedge \svar_{w+1}\}$, then $\{A\}c\{A \wedge \svar_{w+1} \}$.
By implication rule, $A \wedge \svar_{w+1} \Rightarrow aZ_LZ_6 + bX_{2} X_{3} X_{4} X_{8} X_{9}X_{10}$. The next three stabilizer assignment will forward $aZ_LZ_6 + bX_{2} X_{3} X_{4} X_{8} X_{9}X_{10}$. Then with the \textbf{correct} function and the consequence rule, we get that $\{aZ_L + bX_L\}\prog\{aZ_L' + bX'_L\}$, i.e., the logical state is not changed by the qubit moving operation.
\end{proof}

A braiding operation involves many data qubits, and at least 51 data qubits will be referenced in the problem. To simplify the program, we will use the qubit moving as primitive. 
$\textbf{qmov}(X_L, X_L')$ means to move the defect that changes the logical X operation of a X-cut qubit from $X_L$ to $X_L'$, and $\textbf{qmov}(Z_L, Z_L')$ to move the defect that changes the logical Z operation of a Z-cut qubit from $Z_L$ to $Z_L'$.

\begin{program}[Braiding]\label{prog:surf-braid} In the figure below, we braid a Z-cut qubit with a X-cut qubit: \\
\includegraphics[width=0.43\textwidth]{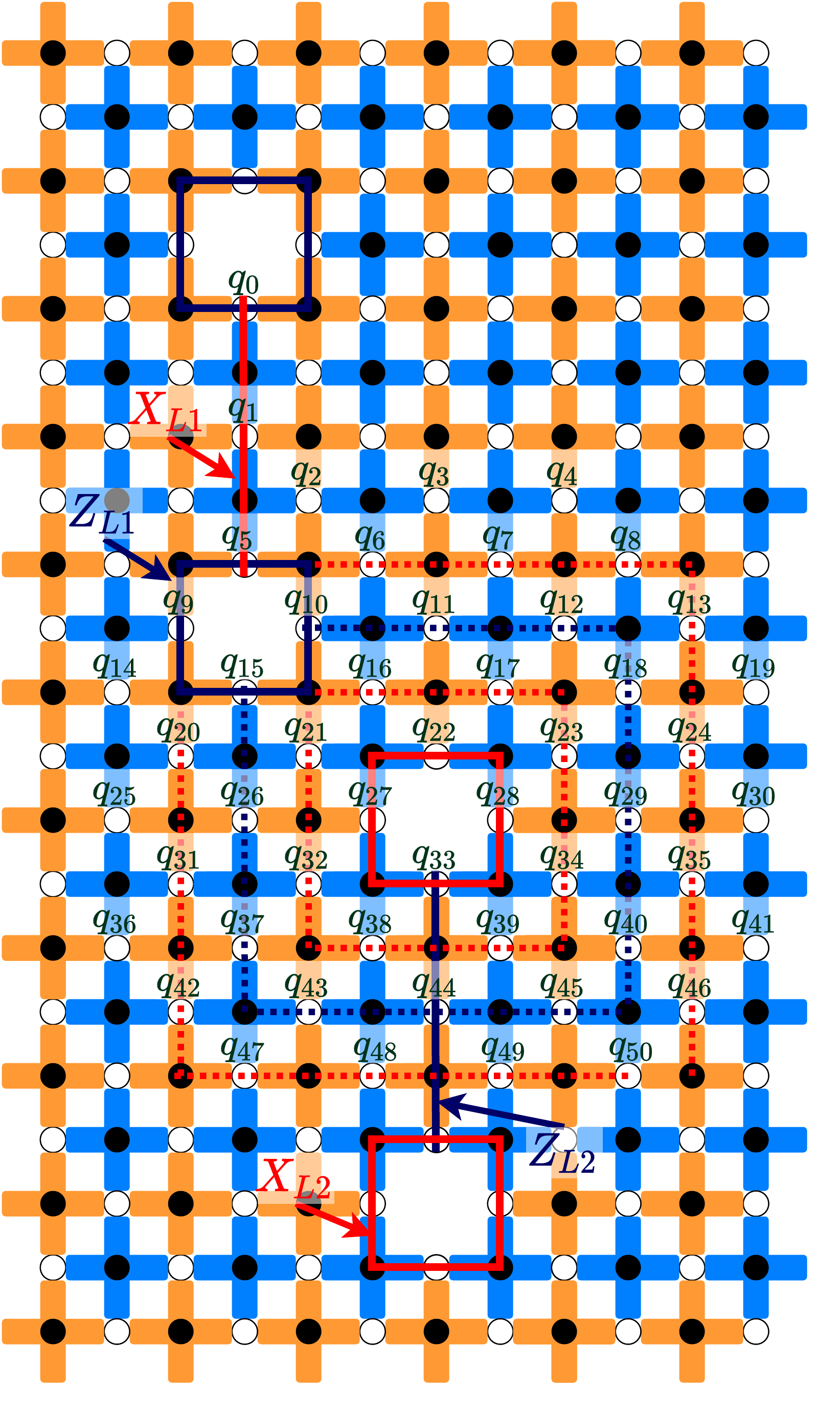}
\\
The associated program is $\prog \Coloneqq$ \\
$
\textbf{qmov}(Z_{5}Z_{9}Z_{10}Z_{15}, Z_{15}Z_{20}Z_{21}Z_{26}) \\
\textbf{qmov}(Z_{15}Z_{20}Z_{21}Z_{26}, Z_{26}Z_{31}Z_{32}Z_{37}) \\
\textbf{qmov}(Z_{26}Z_{31}Z_{32}Z_{37}, Z_{37}Z_{42}Z_{43}Z_{47}) \\
\textbf{qmov}(Z_{37}Z_{42}Z_{43}Z_{47}, Z_{38}Z_{43}Z_{44}Z_{48}) \\
\textbf{qmov}(Z_{38}Z_{43}Z_{44}Z_{48}, Z_{39}Z_{44}Z_{45}Z_{49}) \\
\textbf{qmov}(Z_{39}Z_{44}Z_{45}Z_{49}, Z_{40}Z_{45}Z_{46}Z_{50}) \\
\textbf{qmov}(Z_{40}Z_{45}Z_{46}Z_{50}, Z_{29}Z_{34}Z_{35}Z_{40}) \\
\textbf{qmov}(Z_{29}Z_{34}Z_{35}Z_{40}, Z_{18}Z_{23}Z_{24}Z_{29}) \\
\textbf{qmov}(Z_{18}Z_{23}Z_{24}Z_{29}, Z_{8}Z_{12}Z_{13}Z_{18}) \\
\textbf{qmov}(Z_{8}Z_{12}Z_{13}Z_{18}, Z_{7}Z_{11}Z_{12}Z_{17}) \\
\textbf{qmov}(Z_{7}Z_{11}Z_{12}Z_{17}, Z_{6}Z_{10}Z_{11}Z_{16}) \\
\textbf{qmov}(Z_{6}Z_{10}Z_{11}Z_{16}, Z_{5}Z_{9}Z_{10}Z_{15})
$
\end{program}

According to Fowler, the verification of the braiding operation only need to focus on four configurations of logical states on a pair of logical qubits: $X_{L1}\otimes I_{L2}$, $I_{L1}\otimes X_{L2}$, $I_{L1}\otimes Z_{L2}$ and $Z_{L1}\otimes I_{L2}$.
\begin{restatable}[Braiding]{proposition}{surfbraid} 
For the program $\prog$ in Program~\ref{prog:surf-braid}, \\
$\{X_{L1} I_{L2}\}\prog\{X_{L1}X_{L2}\}$, $\{I_{L1} Z_{L2}\}\prog\{Z_{L1} Z_{L2}\}$, $\{I_{L1}X_{L2}\}\prog\{I_{L1} X_{L2}\}$ and $\{Z_{L1} I_{L2}\}\prog\{Z_{L1} I_{L2}\}$.
\end{restatable}
\begin{proof} To simplify the proof, we let $A_S = \wedge_i \svar_i$, i.e., the assertion generated by current active stabilizers in the surface code array. Note that $A_S$ may change at different time-step. The proof of the braiding operation involves tedious computation and we only give a sketch of the proof here. \\
(1) Prove $\{X_{L1} I_{L2}\}\prog\{X_{L1} X_{L2}\}$. Since $X_{L1} I_{L2} = X_{L1}$, we only need to focus on the reasoning on $X_{L1}$ only. From the verification of the qubit moving, 
$\{X_{L1} I_{L2}\} \textbf{qmov}(Z_{5}Z_{9}Z_{10}Z_{15}, Z_{15}Z_{20}Z_{21}Z_{26}) \{X_{L1}X_{15}  I_{L2}\}$ (after \textbf{correct} function, $\{X_{L1}I_{L2} \wedge X_{15} \wedge A_S$ becomes $\{X_{L1}X_{15}I_{L2} \wedge A_S$). Then, after all these qubit moving operations, we will get \\ $\{X_{L1}I_{L2}\}\prog\{X_{L1} X_{15}X_{26}X_{37}X_{43}X_{44}X_{45}X_{29}X_{18}X_{12}X_{11}X_{10}I_{L2} \wedge A_S\}$. Apply implication rule on $A_S$, we get \\ $A_S \Rightarrow (X_{10}X_{15}X_{16}X_{21}) (X_{21}X_{26}X_{27}X_{32})(X_{32}X_{37}X_{38}X_{43})(X_{33}X_{38}X_{39}X_{44})(X_{34}X_{39}X_{40}X_{45})\\
(X_{23}X_{28}X_{29}X_{34})(X_{12}X_{17}X_{18}X_{23})(X_{11}X_{16}X_{17}X_{22}) = (X_{15}X_{26}X_{37}X_{43}X_{44}X_{45}X_{29}X_{18}X_{12}X_{11}X_{10})(X_{27}X_{33}X_{28}X_{22})
$. Then, by the consequence rule, we have $\{X_{L1}I_{L2}\}\prog\{X_{L1}X_{L2}\}$.
\\
(2) Prove $\{I_{L1} Z_{L2}\}\prog\{Z_{L1} Z_{L2}\}$.
Before the qubit moving operation involves qubits in $Z_{L2}$, the precondition $\{I_{L1}Z_{L2}\}$ will be forwarded by the qubit moving operation. So, we only need to elaborate on $\textbf{qmov}(Z_{38}Z_{43}Z_{44}Z_{48}, Z_{39}Z_{44}Z_{45}Z_{49})$. \\
Before measuring $q_{44}$ in X basis, the assignment statement about $X_{44}$ will turn the precondition $\{I_{L1} Z_{L2}\}$ into \\
$Z_{L2}(Z_{39}Z_{44}Z_{45}Z_{49})$, following the previous verification steps of qubit moving. The if statement on $X_{44}$ and $q_{44}$ will then transform the precondition into $Z_{L2}(Z_{39}Z_{44}Z_{45}Z_{49}) \wedge X_{44}$. The following assignment statement about $Z_{38}Z_{43}Z_{44}Z_{48}$ will turn the precondition $Z_{L2}(Z_{39}Z_{44}Z_{45}Z_{49}) \wedge X_{44}$ into  $Z_{L2}(Z_{39}Z_{44}Z_{45}Z_{49})$. Likewise, the remaining qubit moving operations will change the precondition $Z_{L2}(Z_{39}Z_{44}Z_{45}Z_{49})$ to $Z_{L2}(Z_{40}Z_{45}Z_{46}Z_{50})$, $\cdots$, until $Z_{L2}(Z_{5}Z_{9}Z_{10}Z_{15})$, which is just $Z_{L1}Z_{L2}$. Thus, $\{I_{L1} Z_{L2}\}\prog\{Z_{L1} Z_{L2}\}$.
\\
(3) Prove $\{I_{L1} X_{L2}\}\prog\{I_{L1} X_{L2}\}$. Recall the verification of the qubit moving operation. It is easy to see that $\{I_{L1}\}\textbf{qmov}\{I_{L1}\}$ for any qubit moving operation in $P$. On the other hand, the qubit moving operations in $P$ does not involve any qubits in $X_{L2}$, so precondition $I_{L1}X_{L2}$ will be forwarded by all qubit moving operations, i.e.,
$\{I_{L1} X_{L2}\}\prog\{I_{L1} X_{L2}\}$.
\\
(4) Prove $\{Z_{L1} I_{L2}\}\prog\{Z_{L1}I_{L2}\}$. Since $Z_{L1} I_{L2} = Z_{L1}$, we only focus on the reasoning of $Z_{L1}$ here. It is obvious that starting from $Z_5Z_9Z_{10}Z_{15}$, the logical Z operator finally returns to $Z_5Z_9Z_{10}Z_{15}$ by a series of qubit moving operations. Thus, $\{Z_{L1} I_{L2}\}\prog\{Z_{L1}I_{L2}\}$.
\end{proof}